\documentclass[]{interact}

\usepackage{epstopdf}% To incorporate .eps illustrations using PDFLaTeX, etc.
\usepackage[natbibapa,nodoi]{apacite}% Citation support using apacite.sty. Commands using natbib.sty MUST be deactivated first!
\usepackage{xcolor}
\usepackage{comment}

\usepackage{amssymb}
\usepackage{amsmath}
\usepackage{mathtools}
\usepackage{graphicx}      % include this line if your document contains figures

\usepackage{subfigure}

\usepackage{algorithm}% <=================http://ctan.org/pkg/algorithms
\usepackage{algpseudocode}% <========== http://ctan.org/pkg/algorithmicx
\usepackage{hyperref} % <===============================================

\theoremstyle{plain}% Theorem-like structures provided by amsthm.sty
\newtheorem{theorem}{Theorem}[section]
\newtheorem{lemma}[theorem]{Lemma}

\newtheorem{proposition}[theorem]{Proposition}

\theoremstyle{definition}
\newtheorem{definition}[theorem]{Definition}

\newtheorem{assumption}[theorem]{Assumption}

\theoremstyle{remark}

\begin{document}

%\articletype{ARTICLE TEMPLATE}% Specify the article type or omit as appropriate

\title{A Learning- and Scenario-based MPC Design for Nonlinear Systems in LPV Framework with Safety and Stability Guarantees}

\author{
\name{Yajie Bao\textsuperscript{a}\thanks{Corresponding author: Yajie Bao (yb18054@uga.edu)}, Hossam S. Abbas\textsuperscript{b}, and Javad Mohammadpour Velni\textsuperscript{c}}
\affil{\textsuperscript{a}School of Electrical \& Computer Engineering, The University of Georgia, GA, USA}
\affil{\textsuperscript{b}Institute for Electrical Engineering in Medicine, University of L\"ubeck, L\"ubeck, Germany}
\affil{\textsuperscript{c}Department of Mechanical Engineering, Clemson University, SC, USA}
}

\maketitle

\begin{abstract}
%% Text of abstract
This paper presents a learning- and scenario-based model predictive control (MPC) design approach for systems modeled in the linear parameter-varying (LPV) framework. Using input-output data collected from the system, a state-space LPV model with uncertainty quantification is first learned through the variational Bayesian inference Neural Network (BNN) approach. The learned probabilistic model is assumed to contain the true dynamics of the system with a high probability and is used to generate scenarios that ensure safety for a scenario-based MPC. Moreover, to guarantee stability and enhance the performance of the closed-loop system, a parameter-dependent terminal cost and controller, as well as a terminal robust positive invariant set  are designed. Numerical examples will be used to demonstrate that the proposed control design approach can ensure safety and achieve desired control performance.
\end{abstract}

\begin{keywords}
Safe scenario-based model predictive control; learning-based control design; linear parameter-varying framework; Bayesian neural networks
\end{keywords}

\section{INTRODUCTION}
Model predictive control has been widely used to control a given process while satisfying a set of constraints and found applications in various domains including vehicular technology \citep{9374094} and chemical processes \citep{ellis2017economic}. Furthermore, learning-based model predictive control (L-MPC) has increasingly received interest for control of complex safety-critical systems operating in uncertain and hard-to-model environments \citep{aswani2013provably,koller2018learning,hewing2020learning}. Uncertainties and hard-to-model dynamics of the environments are learned from data for L-MPC to improve the control performance and guarantee constraints satisfaction. However, the statistical nature of learning-based methods brings new challenges including the generalizability of the learned models and the computational complexity involved in the MPC design \citep{mesbah2018stochastic,9304310}.

Suitable and sufficiently accurate model representations of the system dynamics are crucial to MPC performance, and distributional information on the uncertainties reduces the conservativeness of compact uncertainty sets. Gaussian process (GP) regression is a commonly used non-parametric learning method to identify residual model uncertainty, which provides a point-wise approximation of unknown errors with mean and covariance characterization. However, GP suffers from high computational complexity which grows with the number of recorded data points. Additionally, GP-based MPC design faces the challenge of uncertainty propagation (i.e., the propagation of the resulting stochastic state distributions) over the MPC prediction horizon. This problem becomes even more exacerbated when the known nominal part of the system model is nonlinear. To address this, \citet{koller2019learningbased} linearized the nominal part for uncertainty propagation.

Linear parameter-varying (LPV) models use a linear structure to capture time-varying and nonlinear dynamics of a system with system matrices dependent on so-called \textit{scheduling variable(s)}, and this allows developing computationally efficient design methods \citep{hanema2018anticipative}. Nonlinear systems can be embedded in LPV representations \citep{abbas2021lpv}; as an example, linear switching systems and Markov jump linear systems can be viewed as particular cases of LPV systems with the scheduling variables being a switching sequence and a Markov chain, respectively. Moreover, learning-based methods for the global identification of state-space LPV (LPV-SS) models with arbitrary scheduling dependency using input-output data have been developed \citep{bao2020ifac}, and a variational Bayesian inference Neural Network (BNN) approach \citep{bao2020cdc} has been proposed to quantify uncertainties in state-space LPV model identification of nonlinear systems, which provides a posterior density estimation of the system model parameters given an input-output dataset. 

Different LPV-MPC design schemes given system models have been recently surveyed in \citep{morato2020model}. One challenge with the MPC design in the LPV framework lies in the unknown future evolution of the LPV scheduling variables over the prediction horizon. Two main approaches have been considered in the literature to handle this difficulty: min-max MPC formulation over all possible scheduling trajectories \citep{LEE1997763}, and tube-based design, where possible future trajectories are exploited to reduce the uncertainty in the scheduling variables evolution. \citet{hanema2020heterogeneously} proposed a heterogeneously parameterized tube-based MPC approach with recursive feasibility and stability guarantees for LPV systems without considering uncertainty in system models and disturbance. Moreover, \citet{calafiore2013stochastic} proposed a scenario-based MPC for constrained discrete-time LPV models with bounded scheduling dependency and stochastic scheduling variables. However, \citet{calafiore2013stochastic} assumed that the terminal control law associated with the terminal set of the system model is given, which is not practical. Therefore, the existing LPV-MPC approaches are not directly applicable to learning-based LPV models, due to the high complexity of the learned models with arbitrary scheduling dependency and the complex joint uncertainty in the learned models and the scheduling variables.

In this paper, we assume no true system model is available, but input-output data are given. For data-driven LPV-SS model identification using only input-output data, the majority of the current LPV identification methods, including direct prediction-error minimization (PEM) methods, as well as global subspace and realization-based techniques (SID), assume an affine scheduling dependency with known basis functions, which restricts the complexity of a representation \citep{cox2018towards}. \citet{rizvi2018state} used kernelized canonical correlation analysis (KCCA) to estimate the state sequence and then a least-squares support vector machine (LS-SVM) to capture the dependency structure, which suffers from the kernel function selection and computational complexity. The expectation-maximization algorithms estimate states and matrices alternatively \citep{wills2012system}. To simultaneously estimate states and explore LPV model structural dependency, \citet{bao2020ifac} presented an integrated architecture of artificial neural networks (ANNs). However, the aforementioned methods focus on estimating a set of deterministic parameters rather than characterizing the statistical properties of the estimation, which typically produce good models in the sense of minimizing the expected loss. However, the accuracy under a few operating points can be poor, which can later result in a low-performing controller and safety violation. Furthermore, robust control techniques cannot be employed without quantifying the uncertainty of the estimated model. Gaussian process (GP) has been used to quantify model uncertainty but suffers from cubic complexity to data size and assumes joint Gaussian distributions to describe uncertainties \citep{liu2020gaussian}. Instead, BNNs can provide a fast evaluation of uncertainties after training and approximate arbitrary posterior distributions. \citet{bao2020cdc} proposed a BNN training approach based on \citep{bao2020ifac} to quantify the epistemic uncertainty in the ANN model. In this paper, we employ the proposed BNN architecture to identify an LPV model with uncertainty quantification for robust estimation and control purposes.

Moreover, the quantified uncertainties in the learned LPV model and the future scheduling trajectory will be considered simultaneously for L-MPC design. In particular, based on the characterization of the joint uncertainties, we construct tubes of linear models that contain the true system dynamics almost surely. To ensure safety, the system constraints are enforced by considering the worst case in the model tube. Additionally, chance constraints can be handled by adjusting the tube based on the BNN model. Furthermore, we use the expectation of the cost over the possible system trajectories as the cost function. \textit{Scenario-based MPC} (SMPC) is adopted here to approximate analytically intractable evolution of uncertainties and improve online computational efficiency. Several methods of scenario generation for SMPC have been proposed in the literature, including Monte Carlo (MC) sampling methods \citep{shapiro2003monte}, moment matching methods \citep{hoyland2003heuristic}, and even machine learning techniques \citep{defourny2010machine}. Despite these efforts, the existing methods are typically only applied to convex problems and assume full recourse. Since BNNs are evaluated using MC methods, a straightforward approach for scenario generation is to use models drawn from the posterior distributions as scenarios. However, the number of models required for safety guarantees can be too large for online optimization of the SMPC design. To reduce the number of scenarios, \citet{bao2022acc} used $\hat{\mu}_{\mathrm{\hat{y}}(k)}$, $\hat{\mu}_{\mathrm{\hat{y}}(k)}+a^{j}\hat{\sigma}_{\mathrm{\hat{y}}(k)}$, $\hat{\mu}_{\mathrm{\hat{y}}(k)}-a^{j}\hat{\sigma}_{\mathrm{\hat{y}}(k)},j=1,\cdots,\frac{n_{s}-1}{2}$ where $\hat{\mu}_{\mathrm{\hat{y}}(k)}$ and $\hat{\sigma}_{\mathrm{\hat{y}}(k)}$ are the sample mean and standard deviation of the predictions $\mathrm{\hat{y}}(k)$ of uncertainties at time step $k$ by BNNs, and $a^{j}$'s are the tuning multipliers and $n_s$ is the number of scenarios at each node of a stage. However, \citet{bao2022acc} used fixed $a^{j}$'s for each time step, which may not well represent the distributions of uncertainties. In this work, we use K-means \citep{lloyd1982least}, a popular clustering method in machine learning with convergence guarantees, to quantize the sample models. In particular, we apply K-means clustering to possible system matrices at each time step to construct scenarios. Additionally, the distributions of the system matrices are estimated using the identified LPV model and the knowledge of the scheduling variables by Monte Carlo methods. With the scenarios generated using K-means, to maintain the original statistical properties of the system matrices distributions, the probability of the scenarios is estimated by a moment-matching optimization method for matching the first four central moments.

Furthermore, to guarantee the stability of the closed-loop system and the recursive feasibility of the associated MPC optimization problem, we present a learning-based approach for terminal ingredients design. In particular, we transform the BNN model with nonlinear scheduling dependency into an LPV form with affine scheduling dependency and compute a terminal constraint as a robust positive invariant set and a terminal cost as a parameter-dependent poly-quadratic Lyapunov function using related LPV tools \citep{pandey2017quadratic} based on the transformed model. The latter is computed by solving a linear matrix inequality (LMI) problem corresponding to the extreme realizations of the scenarios, which provide a parameter-dependent terminal controller based on the generated scenarios that can improve the control performance. To the best of the authors' knowledge, this is the first work on learning-based terminal control design in the LPV framework that is applicable to general nonlinear systems using only the input-output data. Additionally, the BNN model can be updated online using the framework developed in \citep{bao2020online} with new observations collected by applying the MPC law to the real system. The updated model is anticipated to better characterize the uncertainty of the system, which in turn reduces the conservativeness required to ensure safety and hence improve the control performance. Fig. \ref{fig:flow} shows the flow chart of the overall learning-based SMPC design procedure. 
\begin{figure}
    \centering
    \includegraphics[width=\textwidth]{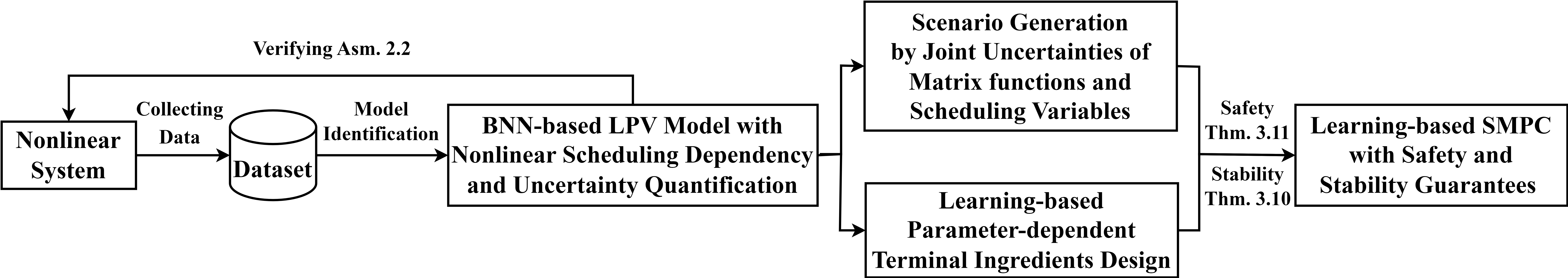}
    \caption{The flow chart of the overall learning-based MPC design procedure.}
    \label{fig:flow}
\end{figure}
 
The contributions of this paper are three-fold: 
\begin{enumerate}
    \item  We present a learning- and scenario-based robust MPC design approach. The proposed approach learns an LPV model with generally nonlinear scheduling dependency from data and thus is applicable to a broader class of nonlinear systems than the existing works that assume a given LPV model with affine scheduling dependency. The proposed approach tackles the challenges of model bias for learning-based control by uncertainty quantification using BNNs and robust control using scenario optimization. The proposed approach can also handle the joint uncertainties of the matrix functions and scheduling variables in the LPV model while most existing LPV-MPC works only consider the uncertainty of scheduling variables.
    \item  We present a learning-based terminal ingredient design for scenario-based MPC using BNN models in the LPV framework. The proposed design reduces conservativeness by considering parameter-dependent terminal ingredients while most existing works consider static terminal ingredients. The proposed design is applicable to LPV models with nonlinear scheduling dependency while most existing works only assume affine scheduling dependency.
    \item  We provide safety and stability guarantees for the proposed MPC scheme.
\end{enumerate}The challenges of the proposed approach lie in learning a sufficiently accurate model from data and reducing the conservativeness of uncertainty quantification and scenario generation for control design purposes. The remainder of the paper is organized as follows: Section 2 describes the problem formulation and preliminaries. Scenario-based MPC design approach using identified BNN models is presented in Section 3. Section 4 then presents numerical results to validate the proposed learning-based control design method. Concluding remarks are finally made in Section 5.

%%%%%%%%%%%%%%%%%%%%%%%%%%%%%%%%%%%%%%%%%%%%%%%%%%%%%%%%%%%%%%%%%

\section{Preliminaries}
\subsection{Basic Definitions}
A set with a non-empty interior that contains the origin is called a proper set, and a proper set that is also compact and convex is called a PC-set. In data-driven methods, a dataset is randomly split into a \textit{training set} for training a model and a \textit{testing set} for testing the generalization of the trained model. In probability theory, an event is said to happen \textit{almost surely} if it happens with probability 1 (or Lebesgue measure 1). A function $f: \mathbb{R}_{+} \xrightarrow[]{} \mathbb{R}_{+}$ is of class $\mathcal{K}_{\infty}$ if it is continuous, strictly increasing, $f (0) = 0$, and $\lim_{\xi\xrightarrow[]{}\infty} f (\xi) = \infty$. A variable $\theta$ is said to evolve according to a bounded rate-of-variation (ROV) if for all time samples $k \in \mathbb{N}$, there exists a $\delta$ such that $|\theta(k+1)-\theta(k)|\leq \delta$. 

\subsection{Problem Formulation}

We consider a constrained discrete-time nonlinear system represented by
\begin{subequations}\label{eq:nonlinear}
\begin{align}
        x(k+1) = f\left(x(k),u(k)\right) \\
        x(k)\in \mathbb{X},~~u(k)\in \mathbb{U}, ~~k\in \mathbb{N},
\end{align}
\end{subequations}
where $f(\cdot)$ is an unknown nonlinear function, $x(k)$ and $u(k)$ denote the states and control inputs at time sample $k$, respectively. $\mathbb{X} \subseteq \mathbb{R}^{n_{x}}$ and $\mathbb{U} \subseteq \mathbb{R}^{n_{u}}$ are the input and state constraint sets.
We can embed the nonlinear representation \eqref{eq:nonlinear} into the following discrete-time state-space LPV representation 
\begin{align}
        x(k+1) = A\left(\theta(k)\right)x(k)+B\left(\theta(k)\right)u(k), \label{eq:lpv} \\
        x(k)\in \mathbb{X},~~u(k)\in \mathbb{U}, ~~k\in \mathbb{N},\label{eq:con} 
\end{align}
where $\theta(k) \in \Theta \subseteq \mathbb{R}^{n_{\theta}}$ denotes the scheduling variables at time sample $k$. The scheduling variables are (nonlinear) functions of inputs/states, but are converted into an exogenous signal by confining the values of $\theta$ to some suitable set $\Theta$ such that the associated set of admissible trajectories (i.e., the set of input and output signals that are compatible with the dynamics) of \eqref{eq:lpv} is a superset of the set of trajectories of the original nonlinear system \eqref{eq:nonlinear} \citep{hanema2018anticipative}. Furthermore, $A$ and $B$ are smooth nonlinear matrix functions of $\theta(k)$. $x(k)$ and $\theta(k)$ can be measured at every time instant $k$. $\mathbb{X}$ and $\mathbb{U}$ are assumed to be PC-sets. Additionally, it is assumed that the future behavior of $\theta$ is not known exactly at time instant $k$ and that the matrix functions $A(\cdot)$ and $B(\cdot)$ are unknown.

Given an initial state $x_{0}$, a scheduling signal $\theta:\mathbb{N}\rightarrow\Theta$, and a control law $\kappa:\mathbb{X}\times\Theta\times\mathbb{N} \rightarrow\mathbb{U}$, the closed-loop system can be described by
\begin{equation}
\label{eq:lpv_closed}
x(k+1)=A(\theta(k))x(k)+B(\theta(k))\kappa(x(k),\theta(k),k)\triangleq\Phi_{\kappa}(x(k),\theta(k),k).
\end{equation}
Additionally, we use $\mathbf{x}(k|\theta,x_{0})$ (resp. $\hat{\mathbf{x}}(k|\theta,x_{0})$) to denote the solution $x(k)$ (resp. $\hat{x}(k)$) to (\ref{eq:lpv_closed}) with the representation (\ref{eq:lpv}) (resp. a data-driven model).
\begin{definition}\label{def:safe}
Given an initial state $x_{0}\in \mathbb{X}$, the system (\ref{eq:lpv}) is said to be \textbf{safe} under a control law $\kappa$ if
\begin{equation}
\label{cond:safe}
\forall k \in \mathbb{N}: \Phi_{\kappa}(x(k),\theta(k),k)\in\mathbb{X},~~~ \kappa(x(k),\theta(k),k)\in \mathbb{U}.
\end{equation}
Moreover, the system (\ref{eq:lpv}) is said to be \textbf{$\delta$-safe} under the control law $\kappa$ if 
\begin{equation}
\label{cond:delta_safe}
\text{Pr}\left[\forall k \in \mathbb{N}:\Phi_{\kappa}(x(k),\theta(k),k)\in\mathbb{X}, \kappa(x(k),\theta(k),k)\in \mathbb{U}\right]\geq \delta
\end{equation}
where $0\leq\delta\leq 1$, and $\text{Pr}[\cdot]$ denotes the probability of an event.
\end{definition}
In general, (\ref{cond:safe}) cannot be enforced without additional assumptions \citep{koller2018learning} especially when (\ref{eq:lpv}) is unknown. Furthermore, $\delta$-safety relaxes the requirements of safety to \textit{safety with a high probability}.
The problem addressed in this paper is to design a learning-based model predictive controller $\kappa: \mathbb{X}\times\Theta\times\mathbb{N}~\xrightarrow[]{}\mathbb{U}$ using a dataset $\mathcal{D}=\{\left(\theta(k),x(k),u(k)\right),x(k+1)\}_{k=1}^{N_{\mathcal{D}}}$ collected from the system, which yields $x(k)~ \xrightarrow[]{} 0$ as $k~\xrightarrow[]{}\infty$ with the constraints \eqref{cond:delta_safe} to be satisfied. First, we briefly describe the proposed probabilistic approach to identify the state-space LPV (LPV-SS) model of the system using the available dataset $\mathcal{D}$. 

\subsection{LPV-SS Model Identification Using BNN}

The data-driven LPV model identification problem is to learn matrix functions $A(\cdot)$ and $B(\cdot)$ from the dataset $\mathcal{D}$. To model arbitrary scheduling dependency and have parametric representations of the system, \citet{bao2020ifac} used fully-connected ANNs to represent $A(\cdot)$ and $B(\cdot)$, and learned the parameters of the ANNs by minimizing the mean squared error (MSE) of the predictions of the ANN model. To quantify the epistemic uncertainty in the ANN model for robust estimation and control, \citet{bao2020cdc} treated the parameters of the matrix functions represented by ANNs as random variables and learned the posterior distributions of the parameters by BNNs \citep{blundell2015weight} composed of DenseVariational layers to represent the matrix functions. 

A BNN approximates the posterior density of the parameters by variational inference given a prior density. In particular, a scaled mixture of two Gaussian densities \citep{blundell2015weight}
\begin{equation}
\label{eq:priors}
    p(w_{j})= \rho_{\text{mix},j}\mathcal{N}(w_{j}|0, \sigma_{j,1}^{2})+(1-\rho_{\text{mix},j})\mathcal{N}(w_{j}|0, \sigma_{j,2}^{2}), 
\end{equation}
with the tuning parameter $\rho_{\text{mix},j}$, is used as the prior density of the parameters $w_{j}$ (including the weights and bias if exists) in the $j$-th layer. Eq. (\ref{eq:priors}) can represent both a heavy tail by a large $\sigma_{j,1}$ and concentration by a small $\sigma _{j,2}$. Furthermore, $\rho_{\text{mix},j}$, $\sigma_{j,1}$, and $\sigma_{j,2}$ are determined using cross validation \citep{hastie2009model}. Variational inference (VI) approximates difficult-to-compute probability density functions by finding a member from a family of densities that is closest to the target in the sense of Kullback–Leibler (KL) divergence \citep{Blei_2017}. To approximate the posterior $p(w_j|\mathcal{D})$, VI solves 
\begin{align}
 ~~~~~~& \min _{\vartheta_{j}} \textup{KL}\Big(q(w _j;\vartheta_{j})\|p(w_j|\mathcal{D})\Big) \label{eq:cst}  \\ 
\Leftrightarrow & ~~~\min _{\vartheta_{j}} \textup{KL}\Big(q(w_j;\vartheta_{j})\|p(w_j)\Big)-\mathbb{E}_{q(w_j;\vartheta_{j})}\left[\log p(\mathcal{D}|w_j) \right] \nonumber \\
\Leftrightarrow  & ~~~\min _{\vartheta_{j}} \Big(\mathbb{E}_{q(w_j;\vartheta_{j})}\left[\log q(w_j;\vartheta_{j})\right]-\mathbb{E}_{q(w_j;\vartheta_{j})}\left[\log p(w_j) \right]  -\mathbb{E}_{q(w_j;\vartheta_{j})}\left[\log p(\mathcal{D}|w_j) \right]\Big), \label{eq:cost}
\end{align}
where $q(w_j;\vartheta_{j})$ denotes a family of densities with parameters $\vartheta_{j}$. The function in (\ref{eq:cost}) is known as the evidence lower bound (ELBO) \citep{Blei_2017}. To solve (\ref{eq:cost}) by Monte Carlo (MC) methods and backpropagation, a reparameterization trick is used to parameterize $q(w_j;\vartheta_{j})$,  i.e., $w_j=\mu_{w{_j}} + \sigma_{w_{j}} \bigodot \epsilon _{w_j}$ where $\bigodot$ denotes the element-wise multiplication, $\epsilon _{w_j} \sim \mathcal{N}(0, I)$, and thus $\vartheta_{j}=(\mu_{w_{j}},\sigma_{w_{j}})$ here. Compared with a Dense layer (i.e., a fully-connected layer with parameters $w_{j}$), a DenseVariational layer (with parameters $\mu_{w_{j}}$ and $\sigma_{w_{j}}$) doubles the number of parameters and requires minimizing ELBO in (\ref{eq:cost}) for uncertainty quantification of $w_{j}$. Similar to ANNs, a BNN can be composed of multiple fully-connected DenseVariational layers. 
\begin{figure}[!htbp]
    \centering
    \includegraphics[width=4in]{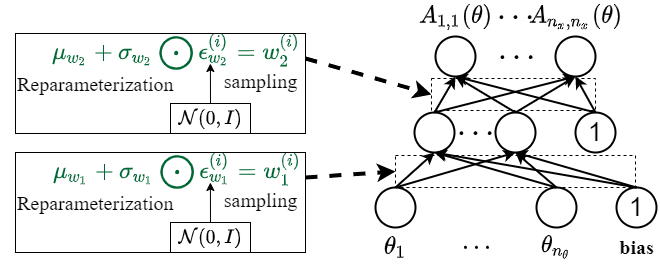}
    \caption{\label{fig:bnn}Using a BNN composed of multiple (here, two) DenseVariational layers to represent $A(\cdot)$ with reparameterization trick. Here, the input to the BNN is $\theta$ and the output is the vectorized $A(\theta)$, which once reshaped, provides the full matrix $A$. BNNs use data to learn the parameters $\mu_{w}$ and $\sigma_{w}$ of the posterior density function.}
\end{figure}

Fig. \ref{fig:bnn} shows how a BNN is used to represent $A(\theta)$; $B(\theta)$ is represented similarly by another BNN. Using $f_{A}^{w}$ and $f_{B}^{w}$ to denote the BNNs representing $A$ and $B$ respectively, the BNN model of the system is described by
\begin{equation}
\label{eq:bnn}
\hat{x}(k+1)=f^{w}(\theta(k),x(k),u(k))=f_{A}^{w}(\theta(k))x(k)+f_{B}^{w}(\theta(k))u(k),
\end{equation}
where $f^{w}$ can be learned by minimizing
\begin{equation}
    \frac{1}{N_\text{BNN}}\sum _{i=1}^{N_\text{BNN}}\left[\log q(w^{(i)};\vartheta)-\log p(w^{(i)})-\log p(\mathcal{D}|w^{(i)}) \right],
\end{equation}
over $\vartheta$ using the dataset $\mathcal{D}$ where $w^{(i)}$ is the $i$-th sample generated by MC for approximating the ELBO, and $N_\text{BNN}$ is the MC sample size determined such that \eqref{eq:bnn} is convergent. Furthermore, as the discussion of the trade-off between bias and variance \citep{bao2020cdc}, using BNNs to represent both $A$ and $B$ increases not only the expressiveness of the LPV model but also the computational cost and compromises the convergence efficiency of the BNN training. It is hence reasonable to only represent $A$ with BNNs and still represent $B$ with ANNs, as $A$ has a larger impact on the system description than other matrix functions. Therefore, in this paper, we consider using BNNs only to represent the matrix $A$, but the proposed approaches can be easily extended to the case where both $A$ and $B$ are represented by BNNs.

Using the trained BNNs, the density of the matrix functions at a given scheduling variable can be evaluated by drawing samples from the posteriors of weights and calculating the possible matrices with each set of sampled weights. Rather than directly estimating the density from samples, we calculate the statistics such as the mean and standard deviation of each element of the matrices, which is efficient and sufficient for constructing a confidence interval of $x(k+1)$ to check \eqref{cond:delta_safe}. The number of samples is determined to guarantee a stable estimation. To provide safety guarantees, we need reliable estimates of the state $x$ inside the operating region $\mathbb{X}\times\mathbb{U}$, which is similar to \citep{bao2022safe} and formally described in the following assumption:
\begin{assumption} \label{assp:lpv}
For a confidence level $\delta_{p} \in (0,1]$, there exists a scaling factor $\beta$ such that with probability greater than $1-\delta_{p}$,
\begin{equation} \label{cond:state}
\forall k \in \mathbb{N}: |x_{j}(k+1) - \hat{\mu} _{x_{j}(k+1)}|\leq \beta _{j} \hat{\sigma} _{x_{j}(k+1)}<|\mathbb{X}_{j}|, ~~~ j = 1, 2, \cdots , n_{x}, 
\end{equation}
given $(x(k), \theta(k), u(k))$ where $\hat{\mu} _{x_{j}(k+1)}$ and $\hat{\sigma} _{x_{j}(k+1)}$ respectively denote the estimated mean and standard deviation of the $j$-th entry of $x(k+1)$ using the learned BNN model with Monte Carlo methods, and $|\mathbb{X}_{j}|$ is used to denote the range of valid $x_{j}$.  
\end{assumption} 

By Assumption \ref{assp:lpv}, the learned model is sufficiently accurate such that the values of $x(k+1)$ of the system are contained in the confidence intervals of our statistical model. It is noted that a larger $\beta_{j}\hat{\sigma}_{x_{j}(k+1)}$ means larger uncertainties of the model and gives a more conservative estimate of $x_{j}(k+1)$ which overestimates the probability of constraints violation, reduces the feasible region of control inputs, and thus degrades control performance. If $\beta_{j}\hat{\sigma} _{x_{j}(k+1)}\geq|\mathbb{X}_{j}|$, the estimate is worse than random guess of $x_{j}(k+1)$, which is not useful for control. The above assumption can be enforced by a well-designed BNN trained on a sufficient dataset and empirically verified on the testing set after model training. Moreover, $\delta_{p}$ can be estimated as the relative frequency of the testing data that violates (\ref{cond:state}) given $\beta$. If Assumption \ref{assp:lpv} does not hold, the architecture of the BNN should be adjusted or more data should be collected for training to improve the accuracy of the BNN until the hypotheses of Assumption \ref{assp:lpv} are satisfied.

\begin{lemma} 
\label{lemma:contain}
Given $x_{0}$, a scheduling signal $\theta$, a BNN model that satisfies Assumption \ref{assp:lpv}, and a confidence level $\delta_{c}$, there exists a scalar $N_{\text{MC}}$ such that 
\begin{equation}
\label{eq:multi}
\begin{split}
&\textup{Pr}\left[\forall k \in \mathbb{N}: \mathbf{x}_{j}(k|\theta, x_{0})\in \left[\min_{i}\hat{\mathbf{x}}_{j}^{(i)}(k|\theta, x_{0}), \max_{i}\hat{\mathbf{x}}_{j}^{(i)}(k|\theta, x_{0}) \right]\right]\geq 1- \delta_{c}, \\
&~~~~~~~~i=1,\cdots, N_{\text{MC}},~~j=1,2,\cdots,n_{x},
\end{split}
\end{equation}
where $N_{\text{MC}}$ is the number of models drawn from the BNN model using MC methods. 
\end{lemma}
\begin{proof}
Let $k=0$, $\hat{x}_{0}=x_{0}$, as $x_{0}$ and $\theta(0)$ are known. Then, using Assumption \ref{assp:lpv}, there exists an $N_{\text{MC}}(0)$ at time step $0$ such that
\begin{equation*}\mathbf{x}_{j}(1|x_{0},\theta(0))\in \left[\min_{1\leq i\leq N_{\text{MC}}(0)}\hat{\mathbf{x}}_{j}^{(i)}(1|x_{0},\theta(0)),\max_{1\leq i\leq N_{\text{MC}}(0)}\hat{\mathbf{x}}_{j}^{(i)}(1|x_{0},\theta(0))\right]
\end{equation*}
almost surely, $j=1,\cdots,n_{x}$. Using induction, (\ref{eq:multi}) is obtained as $N_{\text{MC}}=\max_{k}N_{\text{MC}}(k)$. \end{proof}
It is noted that \textit{almost surely} is used in the proof to avoid the analysis of $\textrm{Pr}$ in (\ref{eq:multi}) which involves the analysis of the closed-loop system and the BNN models and is unnecessary for the proposed approach, although using confidence level can decrease $N_{\text{MC}}$. Lemma \ref{lemma:contain} guarantees that, with a high probability, the system state trajectory is always contained in the multiple trajectories simulated by the BNN model. The uncertainties in the evolution of scheduling variables will be addressed in Section 3.

%%%%%%%%%%%%%%%%%%%%%%%%%%%%%%%%%%%%%%%%%%%%%%%%%%%% 

\subsection{Scenario-based MPC Design Approach}

Given the distribution of the uncertainties described by the BNN model, stochastic MPC can be used to stabilize the system at the origin. In particular, the objective of the stochastic MPC problem is to minimize
\begin{align}
\label{eq:stochastic-mpc}
    \mathbb{E}\left\{\sum_{i=0}^{N-1}\ell(x(i|k),u(i|k))+V_{f}(x(N|k))\right\},
\end{align}
where $\mathbb{E}$ denotes the expected value operator over the random matrix functions and scheduling variables, $\ell(\cdot)$ is the stage cost function, and $V_{f}(\cdot)$ is the terminal cost function. It is noted that the joint uncertainties of matrix functions and scheduling signals are propagated forward through the prediction model \eqref{eq:bnn} over the prediction horizon and thus the closed-form probability density function of $x$ is hard to derive. Therefore, the problem of optimizing \eqref{eq:stochastic-mpc} with the BNN model is not directly solvable. 

Scenario-based MPC (SMPC) assumes that the uncertainty of a system may be represented by a tree of discrete scenarios which facilitates multi-step ahead predictions and feasibility guarantees. As a sufficiently large number of independent uncertainty realization paths by sampling and simulation can represent system uncertainty, applying reduction techniques to the paths can obtain representative scenarios while preserving statistical properties \citep{xu2012scenario} and reduce the computational complexity of SMPC. Any particular branch stemming from a node represents a particular scenario of an unknown, uncertain influence (e.g., from a disturbance or model error) \citep{lucia2013multi}. To represent the trajectories generated by some number $C$ scenarios, we adopt the notation $\left(x^{j}(i), u^{j}(i)\right)$, where the addition of the superscript $j$ indicates the particular scenario $j\in \{1,\ldots,C\}$.

The scenario-based optimal control problem for an uncertain system at time step $k$ can then be formulated as follows
\begin{subequations} \label{eq:ocp-general}
\begin{align}
    \min_{x^{j},u^{j}} &~~ \sum_{j=1}^{C} p^{j} \left[ \sum_{i=0}^{N-1} \ell\left(x^{j}(i|k),u^{j}(i|k)\right) + V_{f}\left(x^{j}(N|k)\right) \right] \label{eq:mpc-cost}\\
    \text{s.t.} &~~ x^{j}(i+1|k)=f_{A}^{w}(\theta^{j}(i|k))x^{j}(i|k)+f_{B}^{w}(\theta^{j}(i|k))u^{j}(i|k), \\
    &~~ \left(x^{j}(i|k),u^{j}(i|k)\right) \in \mathbb{X} \times \mathbb{U}, \\
    &~~ x^{j}(0|k) = x(k), \\
    &~~ u^{j}(i|k) = u^{l}(i|k) ~ \text{if} ~ x^{p(j)}(i|k) = x^{p(l)}(i|k), \label{eq:non-anticipativity}
\end{align}
\end{subequations}
where $p^{j}$ is the probability of the $j$-th scenario, $\ell\left(x^{j}(i|k),u^{j}(i|k)\right)$ is the stage cost, and $V_f\left(x^{j}(N|k)\right)$ is the terminal cost for the trajectory of the $j$-th scenario, $N$ is the prediction horizon, and \eqref{eq:non-anticipativity} enforces a \emph{non-anticipativity} constraint, which represents the fact that each control input that branches from the same parent node must be equal ($x^{p(j)}(i|k)$ is the parent state of $x^{j}(i+1|k)$). The non-anticipativity constraint is crucial to accurately model the real-time decision-making problem such that the control inputs do not anticipate the future (i.e., decisions cannot realize the uncertainty). The solution to this optimization problem is used to generate the control law
\begin{equation}
\label{eq:ctrl-input}
    \kappa\left(x(k)\right) = u^{0*}(0|k).
\end{equation}

Given the structure of the scenario tree, it is crucial to generate appropriate scenarios at each stage of the optimization to accurately represent the uncertainty of the system under consideration. Additionally, the computational cost of SMPC is proportional to the number of scenarios which is positively correlated with the coverage of the uncertainty space. Hence, the objective of constructing a scenario tree is to accurately approximate \eqref{eq:stochastic-mpc} with a relatively small number of scenarios.

To express the joint uncertainties of matrix functions and scheduling signals by scenario trees, we generate model paths by sampling the scheduling signals and simulating the BNN model and apply reduction techniques to the model paths for generating representative scenarios while preserving the statistical properties of uncertainty quantified by the BNN model. In particular, we use MC sampling methods and K-means, a clustering method in machine learning, to generate scenarios, as the uncertainties are described by a BNN model such that the propagation of uncertainties is intractable to analyze. In particular, MC methods are employed to sample models from the BNN model for selected scheduling trajectories. While Lemma \ref{lemma:contain} claims there exists a scalar $N_{\text{MC}}$ such that the trajectories of the sampled $N_{\text{MC}}$ models contain the system trajectory, $N_{\text{MC}}$ can be too large for online optimization of the SMPC problems. Instead, we apply K-means clustering to the $N_{\text{MC}}$ models to reduce the number of scenarios. K-means clustering is a vector quantization method which partitions $N_{\text{s}}$ observations/samples $\{\mathrm{x}^{(i)}\}_{i=1}^{N_{\text{s}}}$ into $C$ disjoint clusters $\{S_{c}\}_{c=1}^{C}$ by minimizing the within-cluster sum-of-squares variances (squared Euclidean distances) $\sum_{c=1}^{C}\sum_{\mathrm{x}\in S_{c}}\|\mathrm{x}-\mu_{c}\|^{2}$, and each cluster is described by the mean (a.k.a. \textit{centroid}) of the samples in the cluster. We use the cluster centroids of the models sampled from the BNN model as the models of scenarios. However, the $C$ scenarios may lose the property of the $N_{\text{MC}}$ models in Lemma \ref{lemma:contain}. 

To incorporate a probabilistic safety certificate into the scenario generation, we add extra scenarios corresponding to the worst cases based on the $N_{\text{MC}}$ models. Then, the system is safe under \eqref{eq:ctrl-input} if \eqref{eq:ocp-general}, where all the scenarios are subject to the constraints, is feasible. Details of the scenario generation with safety guarantees will be provided in the next section.

%%%%%%%%%%%%%%%%%%%%%%%%%%%%%%%%%%%%%%%%%%%%%%%%%%%%%%%%%%%%%%%%%%%%%%%%%%%%%%%%

\section{Scenario-based MPC Design Using the Learned BNN Models}

In this section, we present the techniques employed to design learning-based SMPC for nonlinear systems in the LPV framework with safety and stability guarantees. First, K-means clustering for scenario generation based on the BNN model is presented. Then, the use of the moment-matching method to compute the probability of scenarios is described. Next, the SMPC problem and terminal ingredients are presented, and finally, conditions for the stability and safety guarantees are provided.

\subsection{Proposed Method for Scenario Generation}
\label{sec:scenario_generate}

In this work, we consider both the uncertainty in the evolution of $\theta$ and the epistemic uncertainty from the learning-based modeling. The scenario tree is designed to cover the joint uncertainty space while considering the computational cost. Considering that the matrix functions given a $\theta$ are evaluated using MC methods, we generate multi-stage scenario trees by applying K-means to the models drawn from the BNN model, which is summarized by the following procedure. 
\vspace{2mm}
\begin{algorithmic}[1]
\label{alg:cluster}
\Procedure{Scenarios Generation Using K-means}{}
\State Generate $L$ scheduling trajectories $\{\theta^{(l)}(k),k=1,...,K\}_{l=1}^{L}$ with $K$ time steps using the knowledge of $\theta$. 
\State Evaluate $A(\theta^{(l)}(k)), l=1,...,L, k=1,...,K$ for each time instant and each scheduling trajectory.
\State For each time instant $k$, apply K-means to $\{\text{vec}\left(A(\theta^{(l)}(k))\right)\}_{l=1}^{L}$ to cluster the $L$ evaluations of the matrix function $A$ at time instant $k$ into $C$ clusters.
\State Use the cluster centers as scenarios at time instant $k$.
\EndProcedure
\vspace{2mm}
\end{algorithmic}
It is noted that any knowledge of the scheduling variable can be easily incorporated into the scheduling trajectory generation (line 2 of the above procedure) to reduce the conservativeness of the generated scenarios. When no knowledge except $\Theta$ exists, the future $\theta$ (within a prediction horizon) is assumed to be uniformly distributed over $\Theta$ for scheduling trajectory generation. Additionally, the number of clusters is related to the number of scenarios. Using a larger number $C$ of clusters can better describe the distribution of the matrices and thus improve the control performance but also increases the computational cost of multi-stage MPC \citep{lucia2013multi}.

Moreover, to ensure safety with a given confidence level $\delta$, we add $2$ extra scenarios which correspond to the worst cases and thus have $C+2$ scenarios at each branching node. Specifically, we estimate the mean $\mu_{\text{M}}$ and standard deviation $\sigma_{\text{M}}$ for each element of the identified system matrices $A(\theta)$ and $B(\theta)$ over the range of the scheduling variables and determine $\beta_{\text{M}},\text{M}=A,B$ such that $P\left(x(k+1)\in\mathbb{X}\right) \geq 1-\delta$ when using $\hat{\mu}_{\text{M}}\pm \beta_{\text{M}} \hat{\sigma}_{\text{M}}$ as the worst-case scenarios. Since the elements in the matrices are bounded and the trained BNN is assumed to contain the true dynamics of the system by Assumption \ref{assp:lpv}, there must exist a $\beta_{\text{M}}$ such that the behaviors of the scenarios contain those of the system. A larger $\beta_{\text{M}}$ indicates a more conservative estimation of the uncertainty and can degrade control performance, which is verified by our experiments. Specifically, $\beta_{\text{M}}$ can be determined using probabilistic safety methods for BNNs \citep{wicker2020probabilistic,bao2023learning}. In particular, using $f^{\omega}$ to denote the BNN, probabilistic safety calculates the lower bound of the probability $P_{\text{safe}}(T,S)=P_{\omega\sim q(\omega;\theta)}(\forall \mathrm{x} \in T, f^{\omega}(\mathrm{x})\in S)$ guaranteeing that for all inputs in $T$, the output of the BNN is in the safety set $S$ by estimating the maximal safe sets of weights $H=\{\omega|~\forall \mathrm{x}\in T, f^{\omega}(\mathrm{x})\in S\}$. Additionally, $H$ is approximated by continuously combining safe sets of weights for a given number of iterations using Interval Bound Propagation (IBP) \citep{gowal2018effectiveness} or Linear Bound Propagation (LBP) \citep{zhang2018efficient}. In particular, IBP or LBP propagates the input interval, i.e., $T=[x^L,x^U]$, through the first layer, to find values $z^{(1),L}$ and $z^{(1),U}$ such that $z^{(1)}\in[z^{(1),L},z^{(1),U}]$, and then iteratively propagate the bound through each consecutive layer to obtain an interval in the output, which is guaranteed to contain the network output. In our case, we find $H=\{\beta_{\text{M}}|~\forall \theta \in \Theta, x\in\mathbb{X}, u\in \mathbb{U}, f^{\beta_{\text{M}}}(\theta, x, u)\in\mathbb{X}\}$ such that $P_{\text{safe}}([\Theta; \mathbb{X};\mathbb{U}];\mathbb{X})\geq1-\delta$.

Additionally, given the number of clusters at each time instant, the number of scenarios grows exponentially with respect to the horizon. To maintain the computational tractability, branching is only applied for the first $N_{b}<N$ steps (aka the robust horizon \citep{lucia2013multi}), and the realization of the uncertainty at step $N_{b}$ is used for the remaining $N-N_{b}$ steps, which results in $C^{N_{b}}$ scenarios and $C$ is the number of clusters. Fig. \ref{fig:tree} shows an illustrative example of the scenarios in the robust horizon and prediction horizon. It is noted that the number of scenarios $|r(j)|$ and the matrices  $A_{k}^{\cdot}$ at time $k$ reflect the joint uncertainty of epistemic uncertainty in the matrix functions from the system identification and the unknown evolution of the scheduling variables.

\begin{figure}[!htbp]
    \centering
    \includegraphics[width=0.6\textwidth]{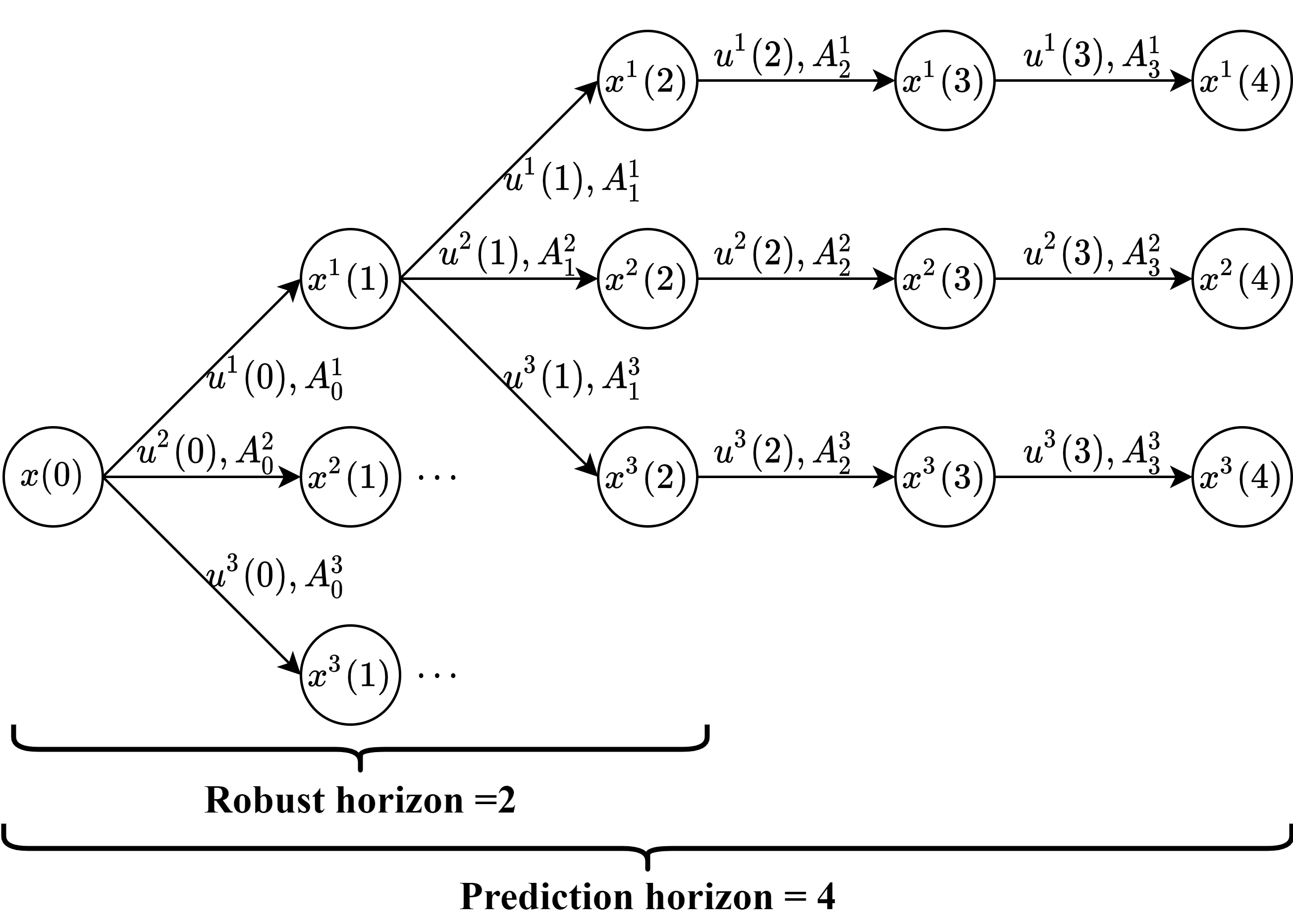}
    \caption{\label{fig:tree} Scenario tree representation of the joint uncertainty evolution for MPC. In the figure, $A_{k}^{r(j)}$ refers to the matrix at time $k$ in the $r(j)$-th scenario.}
\end{figure}

\subsection{Probability of Scenarios} \label{sec:pbs}

After generating the scenario tree, the probability of each scenario is calculated using the moment-matching method \citep{hoyland2001generating} to maintain the original statistical properties. Generally, it is sufficient to use the first four moments as the statistical features to be matched in scenario generation \citep{ji2005stochastic}. Specifically, the first four central moments are matched by solving the following optimization problem
\begin{subequations}
\begin{align}
     \min_{\mathbf{p}} & \sum_{i}^{m}\Big(w_{i}^{1}\left(M_{i}^{-}+M_{i}^{+}\right)+w_{i}^{3}\left(S_{i}^{-}+S_{i}^{+}\right) +w_{i}^{4}\left(Q_{i}^{-}+Q_{i}^{+}\right)\Big) + \\ \nonumber &\sum_{i,j=1}^{m}w_{i,j}^{1}\left(\Sigma_{i,j}^{-}+\Sigma_{i,j}^{+}\right), \\ 
     %\vspace{-2mm}
     \text{s.t.} & ~~~~~~\mathbf{X}\mathbf{p}+M^{-}-M^{+}=M \\
     & \sum _{s=1}^{C+2}(\mathbf{X}^{s}-\mathbf{X}\mathbf{p})(\mathbf{X}^{s}-\mathbf{X}\mathbf{p})^{\mathrm{T}}p^{s}+\Sigma^{-}-\Sigma^{+}=\Sigma, 
     \\
     & \sum _{s=1}^{C+2}(\mathbf{X}^{s}-\mathbf{X}\mathbf{p})^{3}p^{s}+S^{-}-S^{+}=S, \label{eq:3rd}\\
    & \sum _{s=1}^{C+2}(\mathbf{X}^{s}-\mathbf{X}\mathbf{p})^{4}p^{s}+Q^{-}-Q^{+}=Q,\label{eq:4th}\\
    & \sum _{s=1}^{C+2}p^{s}=1, ~~p^{s}\geq0, ~~s=1,\cdots,C, C+1, C+2, \\
    & M_{i}^{+},M_{i}^{-},S_{i}^{+},S_{i}^{-},Q_{i}^{+},Q_{i}^{-}\geq 0, ~~i = 1, \cdots, m, \\
    &\Sigma_{ij}^{+},\Sigma_{ij}^{-}\geq 0,~~i,j = 1, \cdots, m.
    %\vspace{-2mm}
\end{align}
\end{subequations}
where the third- and fourth-power operations in \eqref{eq:3rd} and \eqref{eq:4th} are defined on the elements of the vector $(\mathbf{X}^{s}-\mathbf{X}\mathbf{p})$, $M$, $\Sigma$, $S$, and $Q$ are the first four central moments estimated from samples with superscripts $+,-$ denoting the positive and negative parts of the associated variables, and $w_{i}^{0},w_{ij}^{1}, w_{i}^{3},w_{i}^{4}$ in the objective function are weighting coefficients. Furthermore, $\mathbf{p}=(p^{1},\cdots,p^{C},p^{C+1},p^{C+2})^{\mathrm{T}}$ and $p^{s}$ is the probability of the $s$-th scenario, $\mathbf{X}= (\mathbf{X}^{1},\cdots, \mathbf{X}^{C}, \mathbf{X}^{C+1}, \mathbf{X}^{C+2}) \in\mathbb{R}^{m\times (C+2)}$ where $\mathbf{X}^{s}=(X_{1}^{s},\cdots,X_{m}^{s})$ denotes the realization of the uncertainty in the $s$-th scenario and $m$ is the dimension of the realization. For example, $m=n_{x}^{2}$ when clustering vectorized matrix function value $A$. The optimal value of the cost function is greater than $0$ and indicates the degree to which the generated scenarios preserve the statistical properties of uncertainty quantified by BNNs. Therefore, we choose $C$ such that the optimal value is close to $0$ while satisfying the computational resource limitations of multi-stage MPC.

\subsection{Scenario-based MPC Problem}

Given the constructed tree, the MPC problem can be formulated at every time instant as 
\begin{subequations}
\label{eq:smpc}
\begin{align}
    \min_{x^{j}(i|k),u^{j}(i|k)}& \sum _{j=1}^{C^{N_{b}}}p^{j}J_{j}\left(\mathbf{x}^{j}(0:N|k),\mathbf{u}^{j}(0:N-1|k)\right) \label{eq:obj}\\
    \text{s.t.}~~ & x^{j}(i+1|k)=A_{i}^{r(j)}x^{p(j)}(i|k)+B_{i}^{r(j)}u^{j}(i|k), \\
    & x^{j}(i|k) \in \mathbb{X}, u^{j}(i|k) \in \mathbb{U}, \forall (j,i)\in I,\\ 
    & u^{j}(i|k)=u^{l}(i|k) \text{ if } x^{p(j)}(i|k)=x^{p(l)}(i|k), \forall (j,i),(l,i)\in I,  \label{eq:anti}\\
    & x^{j}(N|k) \in \mathbb{X}_{f}, \forall (j,N)\in I,\label{cons:term}
\end{align}
\end{subequations}
where the weight $p^{j}$ is the probability of the $j$-th scenario computed using the method described in Section \ref{sec:pbs}, the function $p(j)$ refers to the index of the parent node (the parent of the node indexed by $j$), and $r(j)$ gives the considered realization of the joint uncertainty via $A_{i}^{r(j)}$ and $B_{i}^{r(j)}$. Furthermore, $I$ denotes the set of all occurring index pairs $(j,i)$. The constraints in (\ref{eq:anti}) are non-anticipativity constraints to guarantee that control inputs from the same parent node are identical, and $\mathbb{X}_{f}$ is the terminal set. Since each realization of the joint uncertainty gives a linear system model at time instant $i$ in the prediction horizon, the constraint satisfaction can be guaranteed by only considering the extreme scenarios, which is employed to establish safety guarantees. The objective function in (\ref{eq:obj}) is the weighted sum of the cost for each scenario $J_{j}$ which is defined as
\begin{equation}
    J_{j} = V_{f}(x^{j}(N|k)) + \sum_{i=0}^{N-1}\ell(x^{p(j)}(i|k),u^{j}(i|k)),
\end{equation}
in which $V_f(\cdot)$ is the terminal cost and $\ell(\cdot)$ is the stage cost. The terminal cost $V_{f}(\cdot)$ will be discussed in the next subsection. In this paper, we consider 
\begin{equation}
\label{eq:stagecost}
    \ell(x,u)=x^{\mathrm{T}}Qx+u^{\mathrm{T}}Ru
\end{equation}
where $Q,R\succ 0$ are tuning parameters. 

\subsubsection{Terminal Ingredients} \label{sec:terminal} 

In this section, we show how to compute the three terminal ingredients \citep{mayne2000constrained}, i.e., a terminal cost, a terminal controller, and a terminal invariant set, that are required to achieve stability of the closed-loop system with the proposed MPC scheme.

First, we transform the BNN model into an LPV form with affine scheduling dependency described as
\begin{equation}
\label{eq:lpv-affine}
\begin{split}
    &\hat{A}(\hat{\theta}(k)) = \sum _{i=1}^{q}\hat{\theta}_{i}(k)\hat{A}_{i},~~ \hat{B}(\hat{\theta}(k)) = \sum _{i=1}^{q}\hat{\theta}_{i}(k)\hat{B}_{i},\\
    &\sum_{i=1}^{q}\hat{\theta}_{i}(k)=1, ~~\hat{\theta}_{i}(k)\geq 0,
\end{split}
\end{equation}
where $\hat{A}_{i}$ and $\hat{B}_{i}$ are extreme realizations of $A(\theta)$ and $B(\theta)$ in (\ref{eq:lpv}), respectively, and $\hat{\theta}$ is the new scheduling variable such that (\ref{eq:lpv}) and (\ref{eq:lpv-affine}) are equivalent. Additionally, Theorem 2.1 in \citep{Nguyen2014} shows that the LPV models with different numbers of extreme realizations of $A_{i}$ and $B_{i}$ can be transformed into the form of (\ref{eq:lpv-affine}). In particular, we use the scenarios including the worst-case scenarios in Section \ref{sec:scenario_generate} to obtain the extreme realizations of matrix functions $\hat{A}_{i}$ and $\hat{B}_{i}$. It is noted that the conservativeness of the extreme realizations is related to the accuracy of the learned BNN model. Additionally, the number of extreme realizations $\hat{A}_{i}$ is $2^{|A|}$ where $|A|$ denotes the number of elements in matrix $A$, and that number for $\hat{B}_{i}$ is $2^{|B|}$. 
However, we can only measure $\theta$ of the system but not $\hat{\theta}$. Moreover, we assume only input-output data exist without a true system model. Therefore, we learn a coordinate transformation $\mathcal{T}$ from $\theta$ in (\ref{eq:lpv}) to $\hat{\theta}$ in (\ref{eq:lpv-affine}) from data, which can be formulated as a regression problem. While lots of regression algorithms can be used to learn the transformation, ANN can approximate arbitrary nonlinear functions and learn features automatically from data, and thus we use a fully-connected ANN to parameterize $\hat{\theta}(k) =\mathcal{T}(\theta)$ and build the ANN model of $\hat{x}(k+1) =\left( \sum_{i=1}^{q}\hat{\theta}_{i}(k)\hat{A}_{i}\right)x(k)+\left( \sum_{i=1}^{q}\hat{\theta}_{i}(k)\hat{B}_{i}\right)u(k)$ where $\hat{A}_{i},\hat{B}_{i}$ are estimated extreme realizations. Then, the optimal transformation $\mathcal{T}^{\star}$ is obtained by minimizing the Mean Squared Error (MSE) loss function $\frac{1}{N_{\mathcal{D}}-1}\sum_{k=1}^{N_{\mathcal{D}}-1}\left(x(k+1)-\hat{x}(k+1)\right)^{2}$ via stochastic gradient descent (SGD) with respect to the parameters of $\mathcal{T}$ on the dataset $\mathcal{D}$. Additionally, the softmax activation function is used in the last layer to satisfy the constraints of $\hat{\theta}$. The advantage of this approach for coordinate transformation is to further moderate the negative effect of the scenario generation by constraining the scenarios to be compatible with the existing dataset. It is noted that the performance of the proposed approach depends on the sufficiency of the dataset, as well as the architecture design and training of ANNs.

Based on the above formulation, we show how to compute the terminal cost and the related terminal controller. We consider parameter-dependent poly-quadratic terminal cost functions in the form of 
\begin{equation}\label{eq:term_cost}
\begin{split}
    V(x(k),\hat{\theta}(k))&=x(k)^{\mathrm{T}}P(\hat{\theta}(k))x(k),\\
    P(\hat{\theta}(k))&=\sum _{i=1}^{q}\hat{\theta}_{i}(k)P_{i}\succ 0.
\end{split}
\end{equation} 
Note that using such a parameter-dependent formulation can reduce conservativeness significantly in comparison with the parameter-independent counterpart. In general, the closed-loop system can be asymptotically stabilized by the MPC law if there exists a terminal feedback controller $u_{k}=K_{f}(x(k))$ such that the following sufficient conditions are satisfied \citep{mayne2000constrained}:
\begin{enumerate}%[1.]
    \item $V_{f}(\cdot)$ is a Lyapunov function on a terminal set $\mathbb{X}_{f}$ under the terminal controller $K_{f}(\cdot)$ and
    \begin{equation}
    \label{cond:lstage}
        \begin{split}
        &V_{f}(x(k+1))-V_{f}(x(k))\leq - \ell(x(k),K_{f}(x_{k}))<0, \\
        &\forall x(k)\in \mathbb{X}_{f}.
        \end{split}
        \end{equation}
    \item If $x(k)\in \mathbb{X}_{f}$, then $x(k+1)=\hat{A}(\hat{\theta}(k))x(k)+\hat{B}(\hat{\theta}(k))K_{f}(x(k))\in \mathbb{X}_{f}, \forall \hat{\theta}(k)\in\hat{\Theta}$, i.e., $\mathbb{X}_{f}$ is positively invariant under $K_{f}$.
    \item $K_{f}(x)\in \mathbb{U}, \forall x\in \mathbb{X}_{f}\subseteq \mathbb{X}$, i.e., the input and state constraints are satisfied under the control law.
\end{enumerate}
   
\noindent To enlarge the terminal region, we consider the following parameter-dependent state-feedback terminal controller
\begin{equation}
\label{eq:gain}
    K_{f}(x(k);\hat{\theta}(k))=\left(\sum_{i=1}^{q}\hat{\theta}_{i}(k)K_{i}\right)x(k).
\end{equation}
Based on the condition for the stability of discrete-time LPV systems \citep{pandey2017quadratic}, we compute the terminal cost function and controller by the following proposition:

\begin{proposition} For the discrete-time LPV systems described by (\ref{eq:lpv-affine}), condition (\ref{cond:lstage}) is satisfied if there exist matrices $Q_{i}\succ 0, X_{i}\in\mathbb{R}^{n_{x}\times n_{x}}, L_{i}\in \mathbb{R}^{n_{u}\times n_{x}}, Y_{i}\in \mathbb{R}^{n_{u}\times n_{x}}, Z_{i}\in \mathbb{R}^{n_{u}\times n_{x}}, i=1,\cdots,q$ such that 
\begin{equation}
\label{eq:LMI}
\begin{split}
    &\begin{bmatrix}
    X_{i}+X_{i}^{\mathrm{T}}-Q_{i}&X_{i}^{\mathrm{T}}\hat{A}_{i}^{\mathrm{T}}&-L_{i}^{\mathrm{T}}&(Q^{1/2}X_{i})^{\mathrm{T}}& (R^{1/2}L_{i})^{\mathrm{T}} \\
    \star &Q_{j}-R_{i,j} & \hat{B}_{i}Z_{j}-Y_{j}^{\mathrm{T}}&\mathbf{0}&\mathbf{0}\\
    \star &\star &Z_{j}+Z_{j}^{\mathrm{T}}&\mathbf{0}&\mathbf{0}\\
     \star &\star & \star &I &\mathbf{0}\\
     \star &\star & \star & \star &I 
    \end{bmatrix}\succ 0 \\
    &\text{for}~\forall i,j = 1,\cdots q, 
\end{split}
\end{equation}
where \textit{$\star$ represents the symmetric blocks omitted for brevity,} and $R_{i,j}=\hat{B}_{i}Y_{j}+(\hat{B}_{i}Y_{j})^{\mathrm{T}}$, $P_{i}=Q_{i}^{-1}$, using the terminal controller with gain in the form of (\ref{eq:gain}) and $K_{i}=L_{i}X_{i}^{-1}$.
\end{proposition}
\begin{proof}
The proof is based on the proof of Theorem 2 in \citep{pandey2017quadratic}.
Since $X_{i}+X_{i}^{\mathrm{T}}\succ Q_{i}\succ 0$, $X_{i}^{\mathrm{T}}Q_{i}^{-1}X_{i} \succeq X_{i}+X_{i}^{\mathrm{T}} - Q_{i}$. Additionally, substituting  $K_{i}=L_{i}X_{i}^{-1}$ into (\ref{eq:LMI}), we have
\begin{equation}
\label{eq:xi}
    \begin{bmatrix}
    X_{i}^{\mathrm{T}}Q_{i}^{-1}X_{i}&X_{i}^{\mathrm{T}}\hat{A}_{i}^{\mathrm{T}}&-(K_{i}X_{i})^{\mathrm{T}}&(Q^{1/2}X_{i})^{\mathrm{T}}& (R^{1/2}K_{i}X_{i})^{\mathrm{T}} \\
    \star &Q_{j}-R_{i,j} & \hat{B}_{i}Z_{j}-Y_{j}^{\mathrm{T}}&\mathbf{0}&\mathbf{0}\\
    \star &\star &Z_{j}+Z_{j}^{\mathrm{T}}&\mathbf{0}&\mathbf{0}\\
     \star &\star & \star &I &\mathbf{0}\\
     \star &\star & \star & \star &I 
    \end{bmatrix}\succ 0.
\end{equation}
Applying the following congruent transformation \begin{equation*}
   S_{i}^{\mathrm{T}} = \begin{bmatrix}
   X_{i}^{-\mathrm{T}} & \mathbf{0} & \mathbf{0} & \mathbf{0} & \mathbf{0} \\
   \mathbf{0} & I & \mathbf{0} & \mathbf{0} & \mathbf{0} \\
   \mathbf{0} & \mathbf{0} & I & \mathbf{0} & \mathbf{0} \\
   \mathbf{0} & \mathbf{0} & \mathbf{0} & I & \mathbf{0} \\
   \mathbf{0} & \mathbf{0} & \mathbf{0} & \mathbf{0} & I 
   \end{bmatrix} 
\end{equation*}
to (\ref{eq:xi}) gives
\begin{equation}
\label{eq:rewrit}
    \begin{bmatrix}
    Q_{i}^{-1}&\hat{A}_{i}^{\mathrm{T}}&-K_{i}^{\mathrm{T}}&(Q^{1/2})^{\mathrm{T}}& (R^{1/2}K_{i})^{\mathrm{T}} \\
    \star &Q_{j}-R_{i,j} & \hat{B}_{i}Z_{j}-Y_{j}^{\mathrm{T}}&\mathbf{0}&\mathbf{0}\\
    \star &\star &Z_{j}+Z_{j}^{\mathrm{T}}&\mathbf{0}&\mathbf{0}\\
     \star &\star & \star &I &\mathbf{0}\\
     \star &\star & \star & \star &I 
    \end{bmatrix}\succ 0
\end{equation} which can be rewritten as 
\begin{equation}
\label{eq:rewritten}
    \begin{bmatrix}
    P_{i}&\hat{A}_{i}^{\mathrm{T}}&-K_{i}^{\mathrm{T}}&(Q^{1/2})^{\mathrm{T}}& (R^{1/2}K_{i})^{\mathrm{T}} \\
    \star &P_{j}^{-1}+M_{i,j} & \hat{B}_{i}H_{j}^{-\mathrm{T}}-P_{j}^{-\mathrm{T}}F_{j}H_{j}^{-1} &\mathbf{0} &\mathbf{0}\\
    \star &\star &H_{j}^{-\mathrm{T}}+H_{j}^{-1}&\mathbf{0}&\mathbf{0}\\
     \star &\star & \star &I &\mathbf{0}\\
     \star &\star & \star & \star &I 
    \end{bmatrix}\succ 0
\end{equation}by defining $P_{i}=Q_{i}^{-1}, H_{i}=Z_{i}^{-\mathrm{T}},F_{i}=P_{i}Y_{i}^{\mathrm{T}}H_{i}$
and $M_{i,j}= -\hat{B}_{i}H_{j}^{-\mathrm{T}}F_{j}^{\mathrm{T}}P_{j}^{-1}-(\hat{B}_{i}H_{j}^{-\mathrm{T}}F_{j}^{\mathrm{T}}P_{j}^{-1})^{\mathrm{T}}$. Then, applying another congruent transformation
\begin{equation*}
    S_{j}^{\mathrm{T}}=\begin{bmatrix}
   I & \mathbf{0} & \mathbf{0} & \mathbf{0} & \mathbf{0} \\
   \mathbf{0} & \mathbf{0} & H_{j} & \mathbf{0} & \mathbf{0} \\
   \mathbf{0} & P_{j} & F_{j} & \mathbf{0} & \mathbf{0} \\
   \mathbf{0} & \mathbf{0} & \mathbf{0} & I & \mathbf{0}  \\
   \mathbf{0} & \mathbf{0} & \mathbf{0} & \mathbf{0} & I 
    \end{bmatrix}
\end{equation*} to (\ref{eq:rewritten}) produces 
\begin{equation}
\label{eq:convex}
    \begin{bmatrix}
    P_{i}&-(H_{j}K_{i})^{\mathrm{T}}&(P_{j}\hat{A}_{i}-F_{j}K_{i})^{\mathrm{T}}&(Q^{1/2})^{\mathrm{T}}& (R^{1/2}K_{i})^{\mathrm{T}} \\
    \star &H_{j}+H_{j}^{\mathrm{T}} & (P_{j}\hat{B}_{i}+F_{j})^{\mathrm{T}}&\mathbf{0} & \mathbf{0}\\
    \star &\star &P_{j}&\mathbf{0}&\mathbf{0}\\
     \star &\star & \star &I &\mathbf{0}\\
     \star &\star & \star & \star &I 
    \end{bmatrix}\succ 0.
\end{equation}
Taking convex combinations of (\ref{eq:convex}) over $i$ and $j$ gives
\begin{equation}
\label{eq:convexed}
\resizebox{\hsize}{!}{$
    \begin{bmatrix}
    P(\hat{\theta}(k))&-(H(\hat{\theta}(k+1))K(\hat{\theta}(k)))^{\mathrm{T}}&(P(\hat{\theta}(k+1))\hat{A}(\hat{\theta}(k)))-F(\hat{\theta}(k+1))K(\hat{\theta}(k))))^{\mathrm{T}}&(Q^{1/2})^{\mathrm{T}}& (R^{1/2}K(\hat{\theta}(k)))^{\mathrm{T}} \\
    \star &H(\hat{\theta}(k+1))+H(\hat{\theta}(k+1))^{\mathrm{T}} &(P(\hat{\theta}(k+1))\hat{B}(\hat{\theta}(k))+F(\hat{\theta}(k+1)))^{\mathrm{T}} &\mathbf{0}&\mathbf{0}\\
    \star &\star &P(\hat{\theta}(k+1))&\mathbf{0}&\mathbf{0}\\
     \star &\star & \star &I &\mathbf{0}\\
     \star &\star & \star & \star &I 
    \end{bmatrix}\succ 0$}.
\end{equation}
Finally, multiplying (\ref{eq:convexed}) by
\begin{equation*}
    S(\hat{\theta}(k)) = \begin{bmatrix}
    I & K(\hat{\theta}(k))^{\mathrm{T}}&\mathbf{0}&\mathbf{0}&\mathbf{0} \\
    \mathbf{0} & \mathbf{0} & I & \mathbf{0} & \mathbf{0} \\
    \mathbf{0} & \mathbf{0} & \mathbf{0} & I & \mathbf{0} \\
    \mathbf{0} & \mathbf{0} & \mathbf{0} & \mathbf{0} & I
    \end{bmatrix}
\end{equation*}
from the left and by its transpose from the right yields
\begin{equation}
\label{eq:proof}
\begin{split}
    &\begin{bmatrix}
    P(\hat{\theta}(k)) & P(\hat{\theta}(k+1))\hat{A}_{c}(\hat{\theta}(k))^{\mathrm{T}} & (Q^{1/2})^{\mathrm{T}} & (R^{1/2}K(\hat{\theta}(k)))^{\mathrm{T}} \\
    \star & P(\hat{\theta}(k+1)) &\mathbf{0}&\mathbf{0}\\
    \star&\star&I&\mathbf{0}\\
    \star&\star&\mathbf{0}&I
    \end{bmatrix}\succ 0,\\
    & \text{for}~\forall \hat{\theta}(k),\hat{\theta}(k+1) \in \hat{\Theta}
    \end{split}
\end{equation}
where $\hat{A}_{c}(\hat{\theta}(k))=\hat{A}(\hat{\theta}(k))+\hat{B}(\hat{\theta}(k))K(\hat{\theta}(k))$. Finally, it can be shown that (\ref{eq:proof}) is equivalent to (\ref{cond:lstage}) by applying the Schur complement, and this concludes the proof.
\end{proof}

\vspace{2mm}

Next, using the controller determined by solving (\ref{eq:LMI}), we can compute a terminal set as a maximal polyhedral robust positively invariant (RPI) set \citep{Nguyen2014}. The considered input and state constraints are in the form of 
\begin{equation}
    \mathbb{X}=\{x\in \mathbb{R}^{n_{x}}|F_{x}x\leq g_{x}\},\mathbb{U}=\{u\in \mathbb{R}^{n_{u}}|F_{u}x\leq g_{u}\}.
\end{equation}
Different from the Procedure 2.1 in \citep{Nguyen2014}, the state constraints of the system (\ref{eq:lpv-affine}) are 
\begin{equation}
\label{eq:cons}
    x_{c}\in \mathbb{X}_{c}, \mathbb{X}_{c}=\{x\in \mathbb{R}^{n_{x}}|F_{c}x\leq g_{c}\}
\end{equation}
where $F_{c} = \begin{bmatrix}
F_{x}^{\mathrm{T}}&
(F_{u}K_{1})^{\mathrm{T}}&
\cdots&
(F_{u}K_{q})^{\mathrm{T}}
\end{bmatrix}^{\mathrm{T}}$ and $g_c=\begin{bmatrix}
g_{x}^{\mathrm{T}}&
g_{u}^{\mathrm{T}}&
\cdots&
g_{u}^{\mathrm{T}}
\end{bmatrix}^{\mathrm{T}}$, as a parameter-dependent controller is used. Therefore, the number of constraints is increased by $(q-1)qn_{u}$, compared against a parameter-independent controller. Then, using Algorithm \ref{alg:rpi}, we can compute a maximal polyhedral RPI set $\Omega_{\max}$ as the terminal set $\Omega_{f}$.

\begin{algorithm}[H]
\begin{algorithmic}[1]
\Statex \textbf{Input:} $\{\hat{A}_{ci}\}_{i=1}^{q}$, $\mathbb{X}_{c}$ defined in (\ref{eq:cons}).
\Statex \textbf{Output:} The maximal RPI set $\Omega_{\max}$.
\State Set $i=0, F_{0}=F_{c}, g_{0}=g_{c}$ and $\mathbb{X}_{0}=\{x\in \mathbb{R}^{n_{x}}:F_{0}x \leq g_{0} \}$.
\State Set $\mathbb{X}_{1}=\mathbb{X}_{0}$.
\State Eliminate redundant inequalities of the following polytope,
\begin{equation*}
    P=\left\{x\in \mathbb{R}^{n_{x}}: \begin{bmatrix}
    F_{0}\\
    F_{0}\hat{A}_{c1}\\
    \vdots\\
    F_{0}\hat{A}_{cq}
    \end{bmatrix}x\leq \begin{bmatrix}
    g_{0}\\
    g_{0}\\
    \vdots\\
    g_{0}
    \end{bmatrix} \right\}
\end{equation*}
\State Set $\mathbb{X}_{0}=P$ and update consequently the matrices $F_{0}$ and $g_{0}$.
\State If $\mathbb{X}_{0}=\mathbb{X}_{1}$ then stop and set $\Omega=\mathbb{X}_{0}$. Else continue.
\State Set $i=i+1$ and go to step $2$.
\caption{(\citet{gilbert1991linear}, \textbf{Procedure 2.1}) Robustly controlled positively invariant set computation}\label{alg:rpi}
\end{algorithmic}
\end{algorithm}
Furthermore, we can compute the robustly $N$-step controlled positively invariant sets based on the maximal RPI set as the domain of attraction (DOA) using Algorithm \ref{alg:doa}. Different from Procedure 2.3 in \citep{Nguyen2014}, we allow control inputs to be different in Step 2 of Algorithm \ref{alg:doa} for different ($A_{i}, B_{i})$ when computing the expanded set, to enlarge the DOA in Step 2 of Algorithm \ref{alg:doa}. Therefore, the number of decision variables is increased by $q-1$ compared to the parameter-independent case. 

\begin{algorithm}[H]
\begin{algorithmic}[1]
\Statex \textbf{Input:} $\{\hat{A}\}_{i=1}^{q}$, $\{\hat{B}\}_{i=1}^{q}$ and the sets $\mathbb{X}$, $\mathbb{U}$ and $\Omega_{\max}$.
\Statex \textbf{Output:} The $N$-step robustly controlled invariant set $C_{N}$.
\State Set $i=0$ and $C_{0}=\Omega_{\max}$ and let the matrices $F_{0}$, $g_{0}$ be the half-space representation of $C_{0}$, i.e., $C_{0}=\{x\in \mathbb{R}^{n}:F_{0}x\leq g_{0}\}$.
\State Compute the expanded set $P_{i}\subset \mathbb{R}^{n_{x}+n_{u}}$
\begin{equation*}
    P_{i}=\left\{(x,u)\in \mathbb{R}^{n_{x}+n_{u}}: \begin{bmatrix}
    F_{i}(\hat{A}_{1}x+\hat{B}_{1}u_{1})\\
    F_{i}(\hat{A}_{2}x+\hat{B}_{2}u_{2})\\
    \vdots\\
    F_{i}(\hat{A}_{q}x+\hat{B}_{q}u_{q})
    \end{bmatrix}x\leq \begin{bmatrix}
    g_{i}\\
    g_{i}\\
    \vdots\\
    g_{i}
    \end{bmatrix} \right\}
\end{equation*}
\State Compute the projection $P_{i}^{(n)}$ of $P_{i}$ on $\mathbb{R}^{n_{x}}$
\begin{equation*}
    P_{i}^{(n)} = \{x\in\mathbb{R}^{n_{x}}:\exists u\in \mathbb{U}~\text{s.t.}~(x,u)\in P_{i}\}.
\end{equation*}
\State Set $C_{i+1}=P_{i}^{(n)} \cap \mathbb{X}$ and let $F_{i+1},g_{i+1}$ be the half-space representation of $C_{i+1}$, i.e.
\begin{equation*}
    C_{i+1} = \{x\in \mathbb{R}^{n_{x}}:F_{i+1}x\leq g_{i+1}\}.
\end{equation*}
\State If $C_{i+1}=C_{i}$, then stop and set $C_{N}=C_{i}$. Else continue.
\State If $i=N$, then stop else continue.
\State Set $i=i+1$ and go to step 2. 
\caption{(\citet{gilbert1991linear}, \textbf{Procedure 2.3}) Robustly $N$-step controlled invariant set computation}\label{alg:doa}
\end{algorithmic}
\end{algorithm}

%%%%%%%%%%%%%%%%%%%%%%%%%%%%%%%%%%%%%%%%%%%%%%%%%%%%%%%%%%%%%%%%%%

\vspace{3mm}

\subsubsection{Recursive Feasibility, Stability and Safety}
In this section, we establish the recursive feasibility, stability, and safety of the proposed learning-based SMPC scheme.

In particular, the recursive feasibility and stability are established by adopting the work \citep{maiworm2015scenario} which considers a nonlinear discrete-time system represented by 
\begin{equation}
x(k+1)=f(x(k),u(k),p(k)),~\text{s.t.}~ x(k)\in\mathbb{X},~u(k)\in\mathbb{U},~p(k)\in\mathbb{P}
\end{equation}
where $p\in\mathbb{R}^{n_{p}}$ denotes the uncertain parameters and $\mathbb{P}$ is a discrete set of $s$ parameter values, under the following assumptions: 
\begin{assumption}[Continuity] \label{assp:rf1}
The functions $f(x,u,p)$, $\ell(x,u)$ and $V_{f}$ are continuous, with $f(0,0,p)=0$ $\forall p\in \mathbb{P}$, $\ell(0,0)=0$ and $V_{f}(0)=0$.
\end{assumption}
\begin{assumption}[Constraints] \label{assp:rf2}
The sets $\mathbb{X}$ and $\mathbb{X}_{f}\subseteq \mathbb{X}$ are closed, and $\mathbb{U}$ is compact. Each set contains
the origin.
\end{assumption}
Establishing recursive feasibility for SMPC is equivalent to requiring that the terminal state $x^{j}(N)$ of each scenario ends in a common control invariant terminal region $\Omega_{f}$ which ensures that the state stays in $\Omega_{f}$ for all system instances when $x(N)\in \Omega_{f}$.
\begin{assumption}[Common terminal region] \label{assp:rf3}
There exists a common terminal region $\Omega_{f}$ that is control invariant for $x(k+1)=f^{j}(x(k),u(k)),\forall j\in \{1,\cdots,s\}$ with $u\in\mathbb{U}$.
\end{assumption}
\begin{proposition} \label{propo}
Suppose that Assumptions \ref{assp:rf1}, \ref{assp:rf2} and \ref{assp:rf3} hold. Then, the SMPC is recursively feasible.
\end{proposition}
\begin{proof}
The proof is similar to that of Proposition 4 in \citep{maiworm2015scenario} and hence omitted here.
\end{proof}
To establish stability, the following assumptions on the stage and terminal costs are made. 
\begin{assumption}[Basic stability assumption]\label{assp:s1}
For $\forall x \in \Omega_{f}$ and $\forall j\in \{1,\cdots,N_{s}\}$,
\begin{equation}
    \min_{\Tilde{u}(k)\in \mathbb{U}} \{V_{f}^{j}\left(f(x,u,p)\right)+\ell(x,u)|f(x,u,p)\in\Omega_{f}\}\leq V_{f}^{j}(x),
\end{equation}
where $N_{s}=s^{N_{b}}$ denotes the number of scenarios and $V_{f}^{j}$ denotes an individual terminal cost function to the $j$-th scenario, holds for all $p\in\mathcal{P}$.
\end{assumption}
Assumption \ref{assp:s1} ensures the descent property of $V_{f}^{j}$ and implies Assumption \ref{assp:rf3} if $V_{f}^{j}(x)$ is a control Lyapunov function. 
\begin{assumption}[Bounds on stage and terminal costs] \label{assp:s2}
The stage cost $\ell(x,u)$ and the terminal costs $V_{f}^{j}(x)$ satisfy 
\begin{align*}
&\ell(x,u) \geq\alpha_{1} (|x|)~~\forall x\in \Omega_{N},\forall u \in \mathbb{U} \\
&V_{f}^{j}(x)\leq\alpha_{2}^{j}(|x|)~~\forall x \in \Omega_{f}~~\text{and}~~\forall j\in \{1,\cdots,N_{s}\},
\end{align*}
in which $\alpha_{1}(\cdot)$ and $\alpha_{2}^{j}(\cdot)$ are $\mathcal{K}_{\infty}$ functions.
\end{assumption}
Assumptions \ref{assp:s1} and \ref{assp:s2} ensure that the value function is a Lyapunov function for $x(k+1)=f^{j}\left(x(k),\kappa_{N}(x(k))\right),~\forall j\in\{1,\cdots,s\}$ on the domain $C_{N}$. The following lemma and theorem are then given.
\begin{lemma}[SMPC stability \citep{maiworm2015scenario}] \label{lemma:smpc}
Suppose that Assumptions \ref{assp:rf1}-- \ref{assp:s2} hold and that $\Omega_{f}$ contains the origin in its interior. Then, the origin is asymptotically stable with a region of attraction $C_{N}$ for the system $x(k+1)=f^{j}(x(k),\kappa_{N}(x(k)))$ for all $j\in\{1,\cdots,s\}$.
\end{lemma}
Furthermore, the above Lemma \ref{lemma:smpc}, which holds for general nonlinear systems with discrete sets of uncertain parameter values, can be adopted for systems in the LPV form \eqref{eq:lpv-affine} with affine scheduling dependency and continuous set of scheduling variables, resulting from the following lemma.
\begin{lemma} \label{lemma:continuous}
$\forall k\in\mathbb{N}$, if a control input $u(k)$ is feasible for all the extreme realizations of \eqref{eq:lpv-affine} given $x(k)$, then $u(k)$ is feasible $\forall \hat{\theta}(k)\in\{\hat{\theta}|\sum_{i=1}^{q}\hat{\theta}_{i}=1$ and $\hat{\theta}_{i}\geq 0\}$ in \eqref{eq:lpv-affine}.
\end{lemma}
\begin{proof}
Since $\hat{A}_{i}x(k)+\hat{B}_{i}u(k)\triangleq x^{i}(k+1) \in \mathbb{X},i=1,\cdots,q$, and $\mathbb{X}$ is assumed to be a PC-set, then $\forall \hat{\theta}(k)\in\{\hat{\theta}|\sum_{i=1}^{q}\hat{\theta}_{i}=1$ and $\hat{\theta}_{i}\geq 0\}, x(k+1)=\hat{A}(\hat{\theta}(k))x(k)+\hat{B}(\hat{\theta}(k))u(k)=\left(\sum_{i=1}^{q}\hat{\theta}_{i}(k)\hat{A}_{i}\right)x(k)+\left(\sum_{i=1}^{q}\hat{\theta}_{i}(k)\hat{B}_{i}\right)u(k)=\sum_{i=1}^{q}\hat{\theta}_{i}(k)(\hat{A}_{i}x(k)+\hat{B}_{i}u(k))=\sum_{i=1}^{q}\hat{\theta}_{i}(k)x^{i}(k+1)\in \mathbb{X}$. 
\end{proof} 
Lemma \ref{lemma:continuous} shows that the feasibility of a control input for all the possible values of scheduling variables can be established by only considering the finite extreme realizations of \eqref{eq:lpv-affine}.

Based on the stability theorem of SMPC and Lemma \ref{lemma:continuous}, we present the following theorem on the learning-based SMPC.
\begin{theorem}[Learning-based SMPC stability and feasibility]
Suppose that Lemma 1 is fulfilled, and the terminal set $\Omega_{f}$ computed by Algorithm \ref{alg:rpi} contains the origin in its interior. Then, the SMPC with the terminal cost \eqref{eq:term_cost} and the terminal controller \eqref{eq:gain} is recursively feasible and the origin is asymptotically stable for \eqref{eq:lpv-affine} with a region of attraction $C_{N}$. Moreover, the original system \eqref{eq:lpv} is stable with a high probability that is at least $1-\delta_{c}$. 
\end{theorem}
\begin{proof}
Obviously, the LPV model with \eqref{eq:lpv-affine}, the considered stage cost \eqref{eq:stagecost}, and the terminal cost \eqref{eq:term_cost} fulfill Assumption \ref{assp:rf1}. Assumption \ref{assp:rf2} also holds, as the terminal set $\Omega_{f}$ computed in Section \ref{sec:terminal} is polyhedral and thus closed while $\mathbb{X}$ and $\mathbb{U}$ are assumed to be PC-sets. Moreover, $\Omega_{f}$ by Algorithm \ref{alg:rpi} is control invariant for arbitrary scheduling variables under the terminal controller \eqref{eq:gain} and thus Assumption \ref{assp:rf3} holds. Furthermore, Assumption \ref{assp:s1} holds, as the designed terminal controller satisfies the sufficient conditions in Section \ref{sec:terminal}. Additionally, the quadratic stage cost \eqref{eq:stagecost} and the poly-quadratic terminal cost function \eqref{eq:term_cost} satisfy Assumption \ref{assp:s2} with $\alpha_{1}(|x|) = \lambda_{\min}(Q)\|x\|^{2}$ and $\alpha_{2}^{j}(|x|)=\lambda_{\max}(P_{j})\|x\|^{2}$ where $\lambda$ denotes the eigenvalue of a matrix. Hence, the origin is asymptotically stable with a region of attraction $C_{N}$ for the LPV model with \eqref{eq:lpv-affine} by Lemma \ref{lemma:smpc}. Moreover, the LPV model with \eqref{eq:lpv-affine} is transformed from the BNN model whose behaviors contain the behaviors of the system by Lemma \ref{lemma:contain}. Therefore, the system is stabilized with a high probability that Lemma \ref{lemma:contain} is fulfilled.
\end{proof}

Furthermore, using the scenario generation approach described in Sections \ref{sec:scenario_generate} and \ref{sec:pbs}, the certificate of safety under the SMPC law can be formalized as follows.
\begin{theorem}[Learning-based SMPC safety]
Let the hypotheses of Assumption \ref{assp:lpv} and Lemma \ref{lemma:contain} be satisfied. Then, the system under the SMPC law \eqref{eq:ctrl-input} is $\delta$-safe. 
\end{theorem}
\begin{proof}
By Assumption \ref{assp:lpv} and Lemma \ref{lemma:contain},  the behaviors of the generated scenarios based on the $N_{\text{MC}}$ sampled models from the BNN model contain the behaviors of the system. Furthermore, by Proposition \ref{propo}, the SMPC is recursively feasible, which proves the system is $\delta$-safe by Definition \ref{def:safe}.
\end{proof}

Additionally, Fig. \ref{fig:block} shows the block diagram of the closed-loop learning-based SMPC scheme.
\begin{figure}
    \centering
    \includegraphics[width=0.7\textwidth]{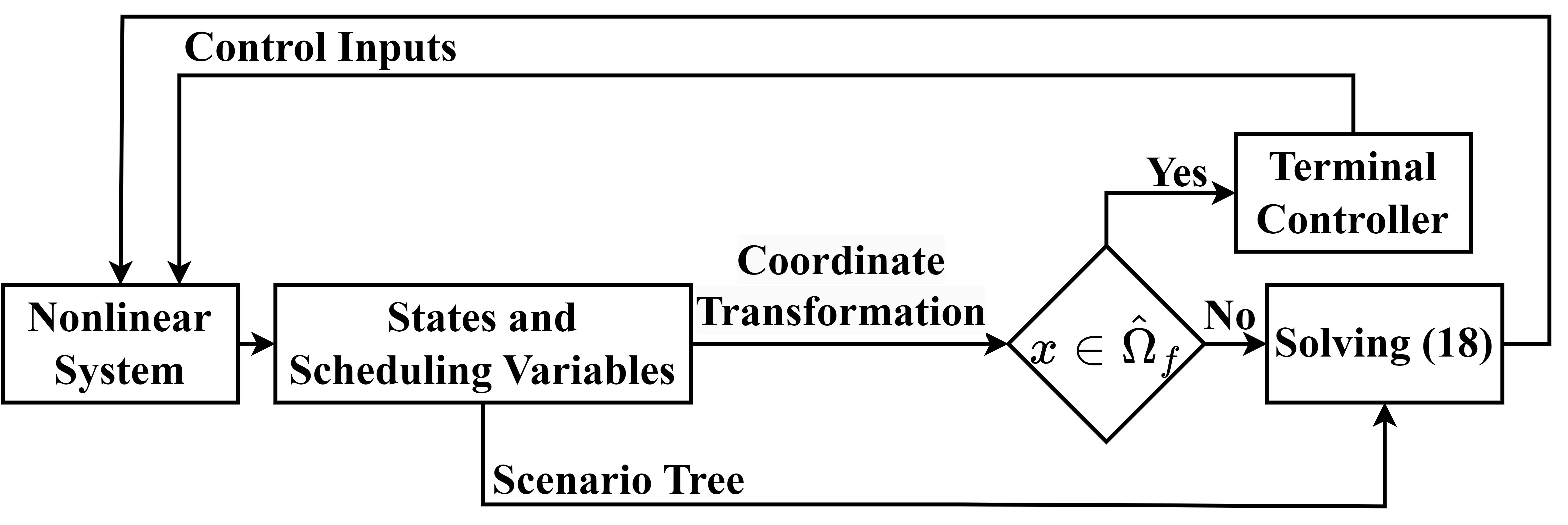}
    \caption{The block diagram of the closed-loop learning-based SMPC scheme.}
    \label{fig:block}
\end{figure}
%%%%%%%%%%%%%%%%%%%%%%%%%%%%%%%%%%%%%%%%%%%%%%%%%%

\section{Numerical Results}

In this section, the proposed methods of this work are validated on a parameter-varying double integrator system model \citep{hanema2020heterogeneously}, as well as a parameter-varying multiple-input multiple-output (MIMO) system with complex nonlinear scheduling dependency. 

\subsection{Parameter-varying Double Integrator}

The LPV-SS representation of the system is assumed to be
\begin{equation}
\label{eq:lpve}
    \begin{split}
    x(k+1)&=\left( 
\begin{bmatrix}
1 & 1 \\
0 & 1 
\end{bmatrix} +  \begin{bmatrix}
0.1 & 0 \\
0 & 0.1 
\end{bmatrix}\theta_{1}(k) + \begin{bmatrix}
0.5 & 0.5 \\
0 & 0 
\end{bmatrix}\theta_{2}(k) \right. \\ & \left. + \begin{bmatrix}
0 & 0 \\
0 & 0.2 
\end{bmatrix}\theta_{3}(k)\right)x(k)+ \begin{bmatrix}
0.5  \\
1  
\end{bmatrix}u(k),
    \end{split}
\end{equation}
with constraints and scheduling sets as
\begin{align*}
    &\mathbb{X}  =\{x\in\mathbb{R}^{2}|\|x\|_{\infty}\leq 6\},
    \mathbb{U}  = \{u\in\mathbb{R}|\left|u\right|\leq 1\} \\
    &\Theta  = \{\theta \in \mathbb{R}^{3}|\|\theta\|_{\infty}\leq 1\}.
\end{align*}
In (\ref{eq:lpve}), $A(\cdot)$ is an affine function of the scheduling variables and $B$ is constant.

\subsubsection{System Identification} 

We use slowly-varying trajectories for the scheduling variables in Fig. \ref{fig:data}(a) to collect observations $\mathcal{D}=\{(\theta (t), x(t),u(t)), x(t+1)\}$ for model identification. Pseudo random binary sequences (PRBS) input signal with a scale of $0.01$ as shown in Fig. \ref{fig:data}(b) is used to excite the system, and the generated state sequence with initial state $x(0)=[2.7;0]$ is shown in Fig. \ref{fig:data}(c), (d). Furthermore, $500$ samples are collected and split into training and testing sets with a ratio of 80\%/20\%.

\begin{figure}[!htbp]
\centering
    \subfigure[Scheduling trajectories.]{\includegraphics[width=0.48\textwidth]{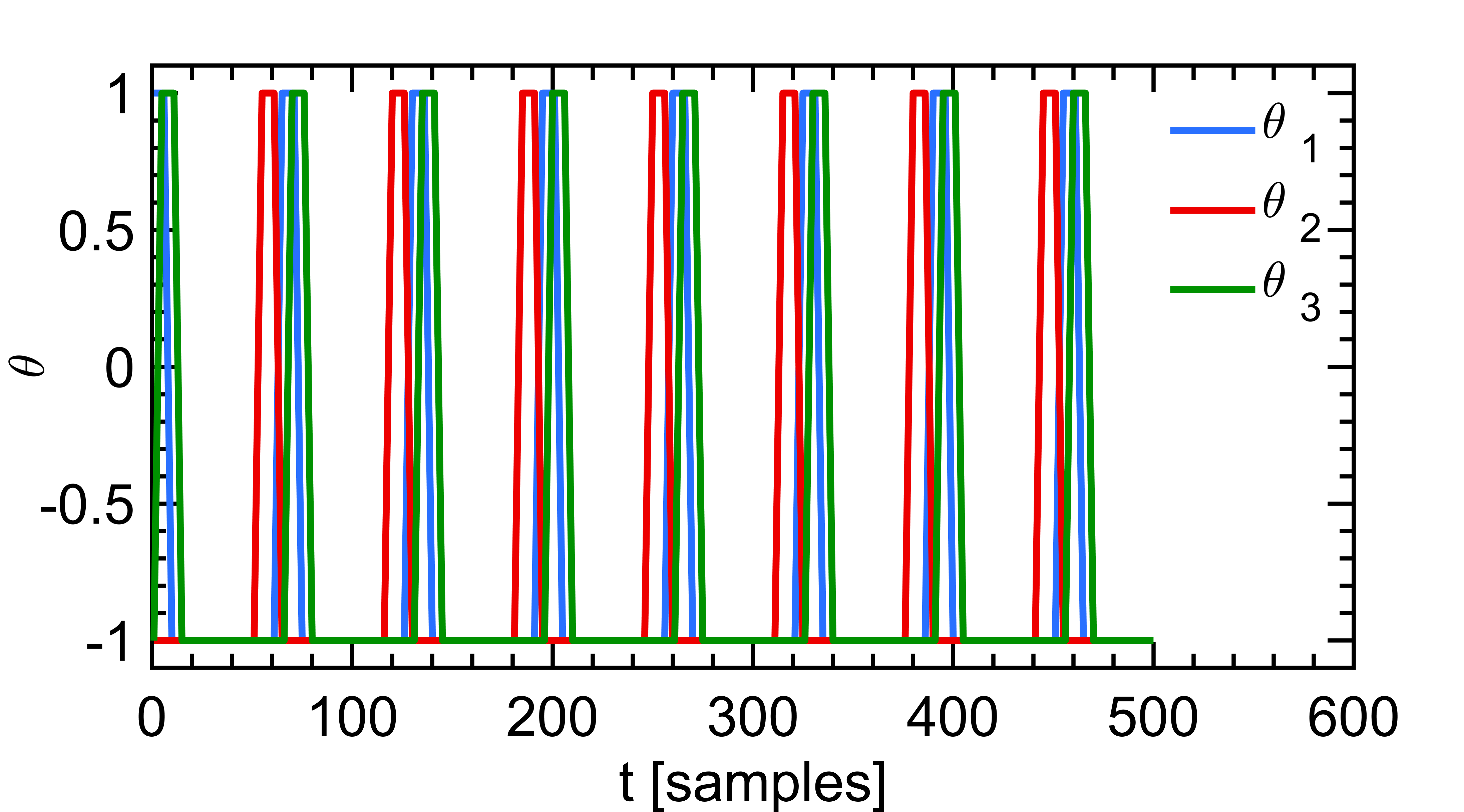}}\label{fig:sched}
    \subfigure[Inputs to the system.]{\includegraphics[width=0.48\textwidth]{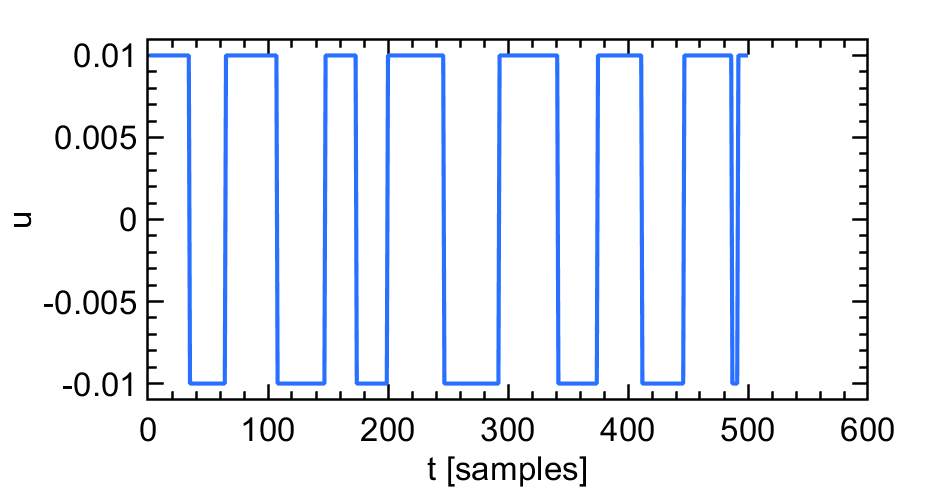}}\label{fig:inp}
    \subfigure[Sequence of $x _{1}$.]{\includegraphics[width=0.48\textwidth]{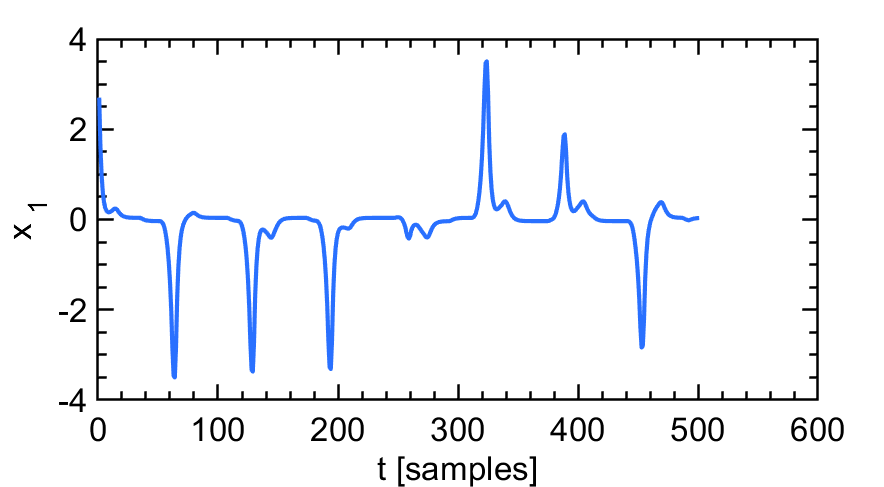}}\label{fig:s1}
    \subfigure[Sequence of $x _{2}$.]{\includegraphics[width=0.48\textwidth]{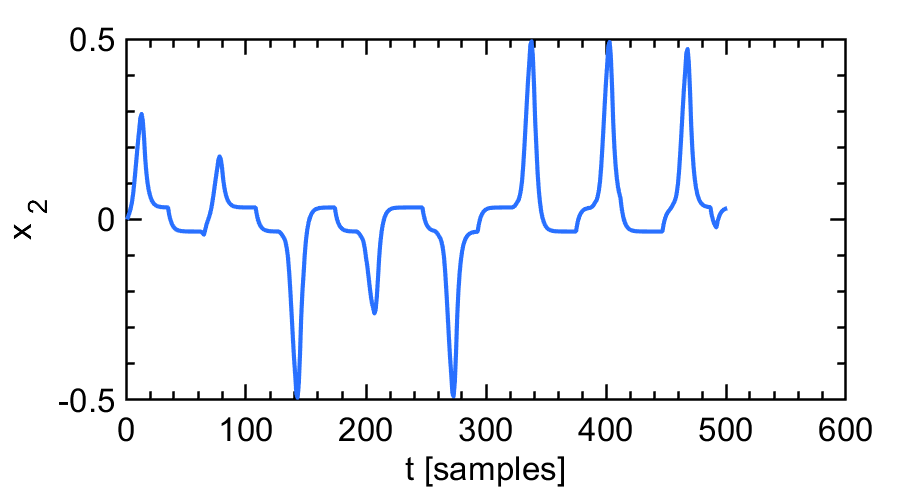}}\label{fig:s2}
    \caption{\label{fig:data}Data generated for model identification purposes.}
\end{figure}

We use one DenseVariational layer with $4$ hidden units to represent $A(\cdot)$ and one Dense layer with $2$ hidden units to represent $B$. Neither of the layers uses activation functions and the Dense layer further does not use bias, which aims to exactly represent the class of models to which (\ref{eq:lpve}) belongs. The tuning parameters in (\ref{eq:priors}) are determined as $\sigma_{1}=0.3, \sigma_{2}=0.1$. Adam optimizer is used with a learning rate set to $0.01$ and other hyper-parameters as default. Moreover, using the transfer learning approach \citep{bao2020cdc}, we first trained an ANN model with the same architecture as the BNN model, used the trained ANN weights to initialize the BNN model, and then trained the BNN model for $1,000$ epochs. The validation results are shown in Fig. \ref{fig:val_1}. 

\begin{figure}[!htbp]
\centering
    \subfigure{\includegraphics[width=0.48\textwidth]{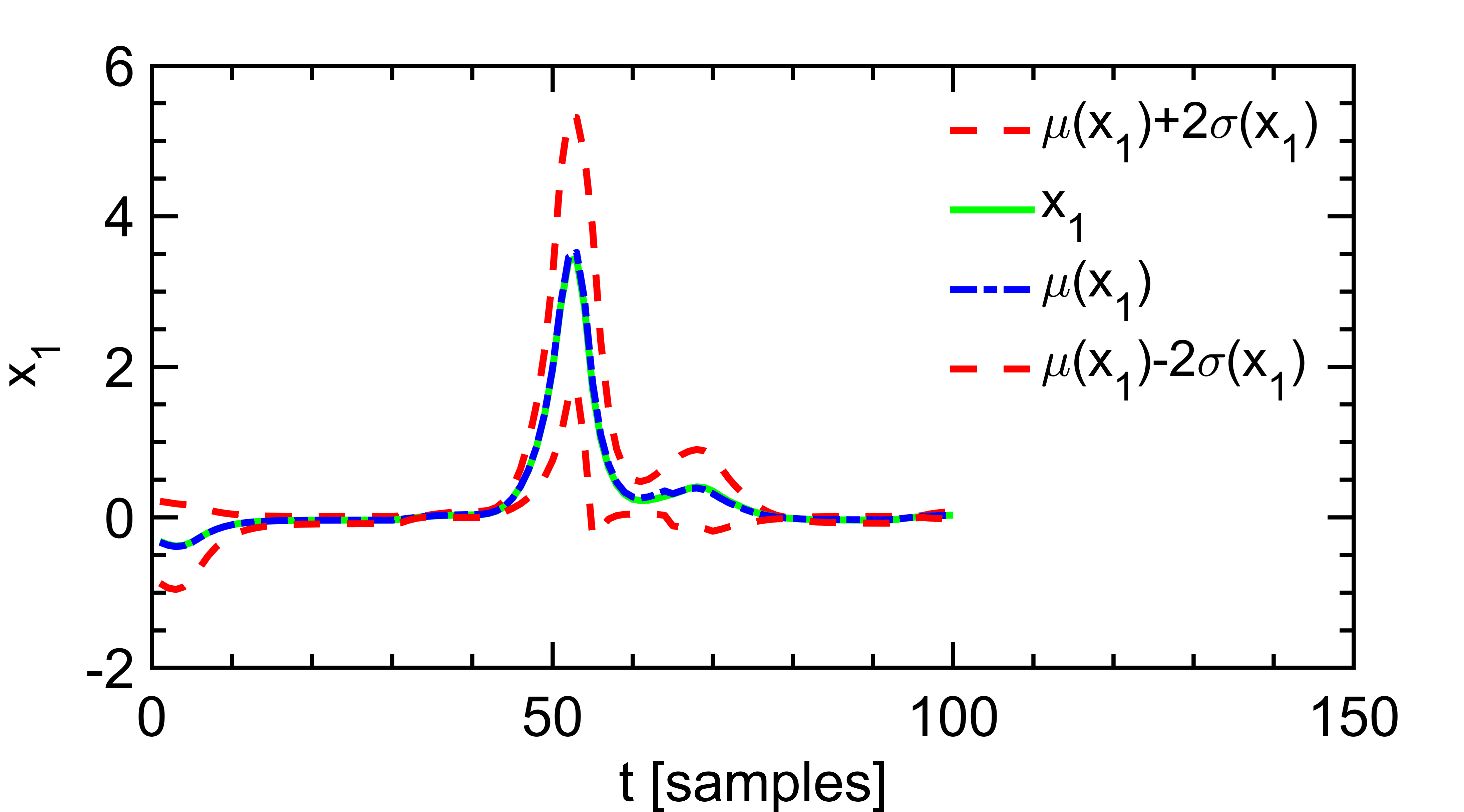}}
    \subfigure{\includegraphics[width=0.48\textwidth]{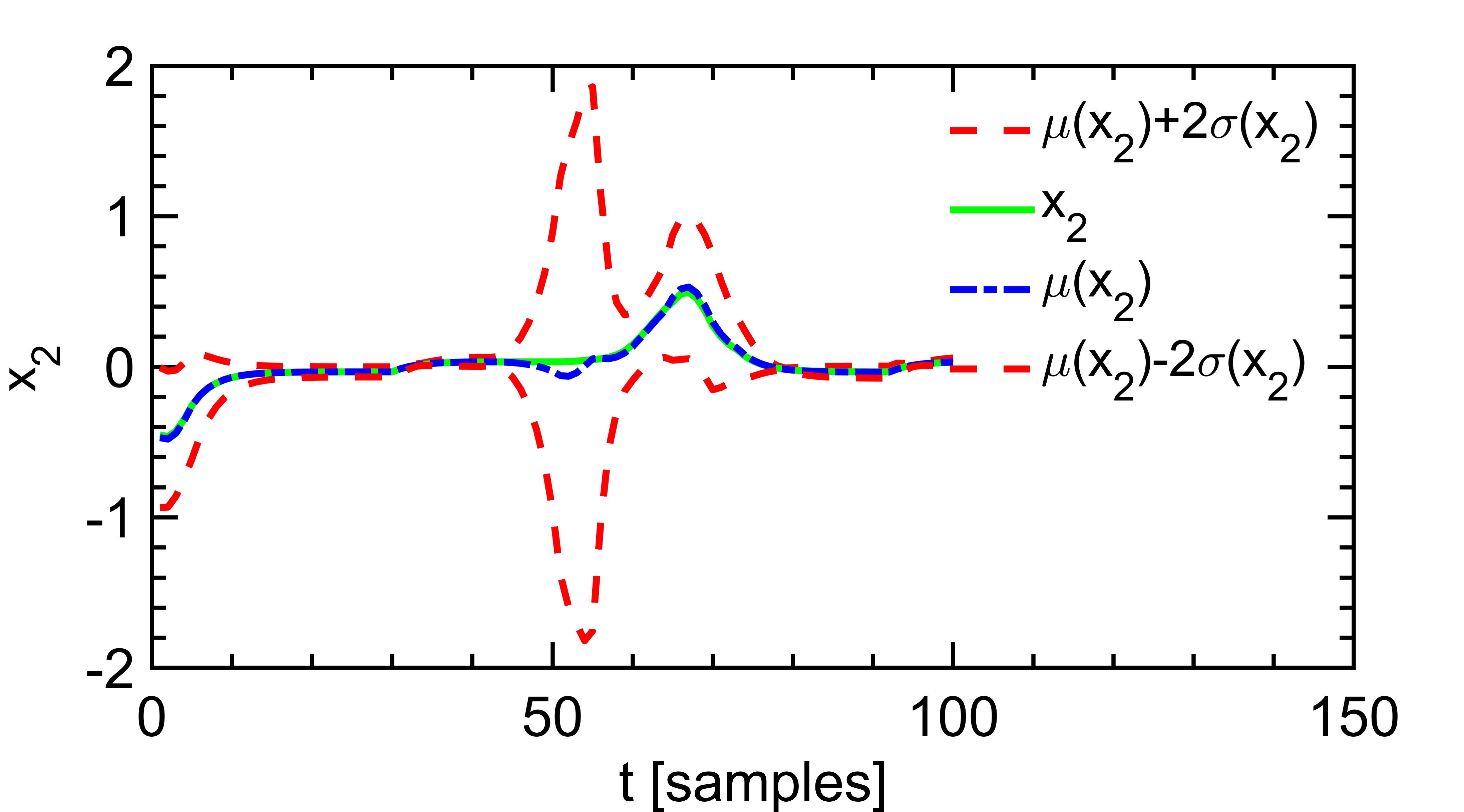}}
    \caption{\label{fig:val_1}Validation results for the identified BNN model. The area between the two dashed lines is within $2$ estimated standard deviations of the estimated mean, which is about $95\%$ confidence interval.}
\end{figure}

It is noted that the best fit ratio $\textup{BFR}=100\%\cdot \max\left(1-\frac{\|x-\hat{x}\|_{2}}{\|x-\bar{x}\|_{2}},0\right)=[96.70\%;87.04\%]$ using the estimated mean as predictions for outputs. None of the samples are out of $2\sigma_{x}$. By increasing $\beta\sigma_{x}$, the true states are guaranteed to lie in the interval $[\mu_{x}-\beta\sigma_{x}, \mu_{x}+\beta\sigma_{x}]$ almost surely.

\subsubsection{Validation of The Proposed Approach}

Without extra knowledge on the evolution of the scheduling variables beyond the scheduling sets, we randomly sample $500$ $\theta$'s from the uniform distribution over $\Theta$ and evaluate $A(\cdot)$ for $N_{\text{MC}}=500$ times using the dynamic functions sampled from the BNN model for each $\theta$. Then, we apply K-means to the evaluated $A$'s to generate the scenarios. The number of clusters is assumed to be 3. Also, $\beta_{\text{M}}=1, \text{M}=A,B$ is considered here for the estimation of extreme realizations. Therefore, 5 scenarios were used including $\mu_{\text{M}}\pm \beta_{\text{M}}\sigma_{\text{M}}$. It is worth noting that the scenarios are fixed within the robust horizon of the tree generation in this case due to the limited time-invariant knowledge of $\theta$. When further information (e.g., a bounded ROV \citep{casavola2008predictive}) is known, we can generate time-varying scenarios for each step within the robust horizon. The probability of the 5 scenarios is $\mathbf{p}=[0.3560; 0.3383; 0.3037; 0.0010; 0.0010]$ using the moment matching method. Additionally, in our experiments, $Q=I_{2\times 2},R=1$ for the stage cost $\ell$ in (\ref{eq:stagecost}). The prediction horizon is set to 10 and the robust horizon to 1. We computed RPI sets and 10-step robustly controlled positively invariant (RCPI) sets based on the system model (\ref{eq:lpv}) and the BNN model (\ref{eq:lpv-affine}), respectively. 

\begin{figure}[H]
    \centering
    \includegraphics[width=0.45\textwidth]{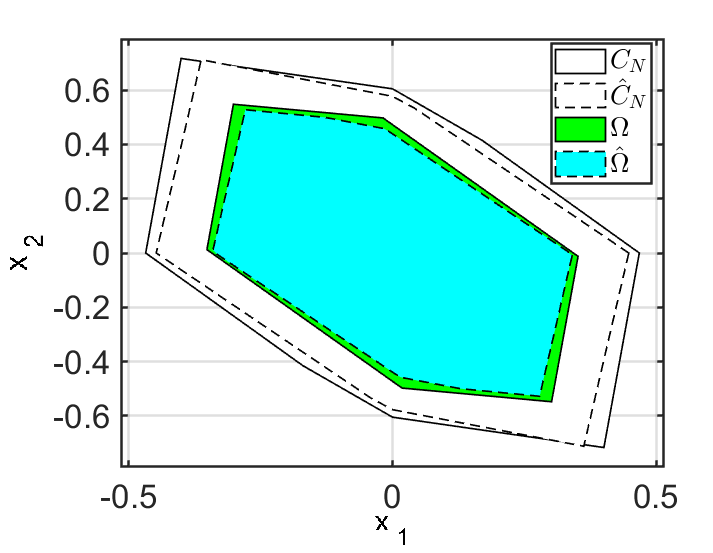}
    \caption{\label{fig:rpi} RPI sets: $\Omega_{f}$ for system and $\hat{\Omega}_{f}$ for BNN model; 10-step RCPI sets: $C_{N}$ for system and $\hat{C}_{N}$ for BNN model.}

\end{figure}

\textbf{Results and Discussion:} As shown in Fig. \ref{fig:rpi}, the estimated sets are smaller than the system sets due to the conservativeness introduced to guarantee safety. The sets can be enlarged by numerical methods, which will be investigated in the future work.  

\begin{figure}[!htbp]
\centering
    \subfigure[Scheduling signals used for control.]{\includegraphics[width=0.48\textwidth]{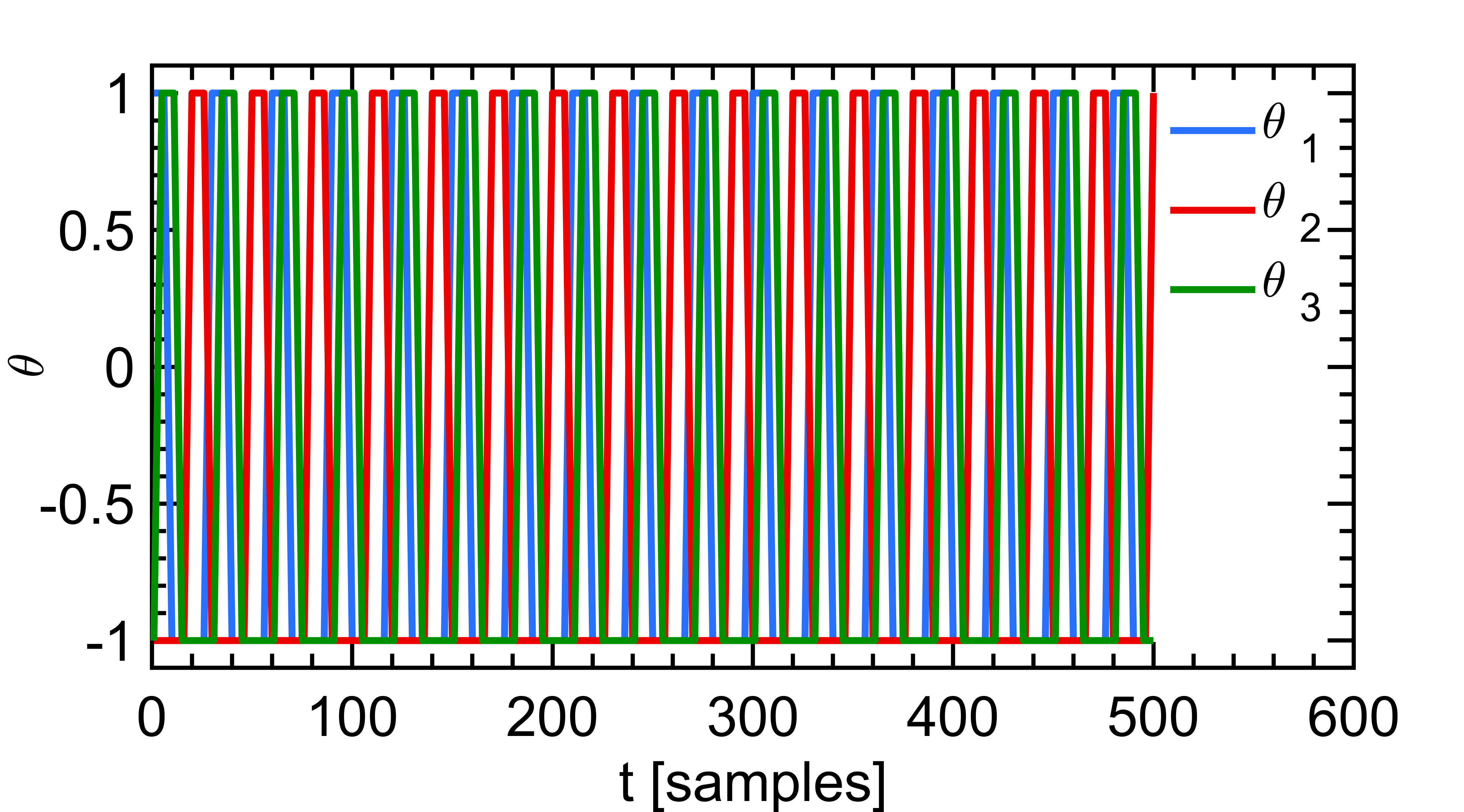}}\label{fig:schedCC}
    \vfill
    \subfigure[Control results without using terminal cost and terminal set.]{\includegraphics[width=0.48\textwidth]{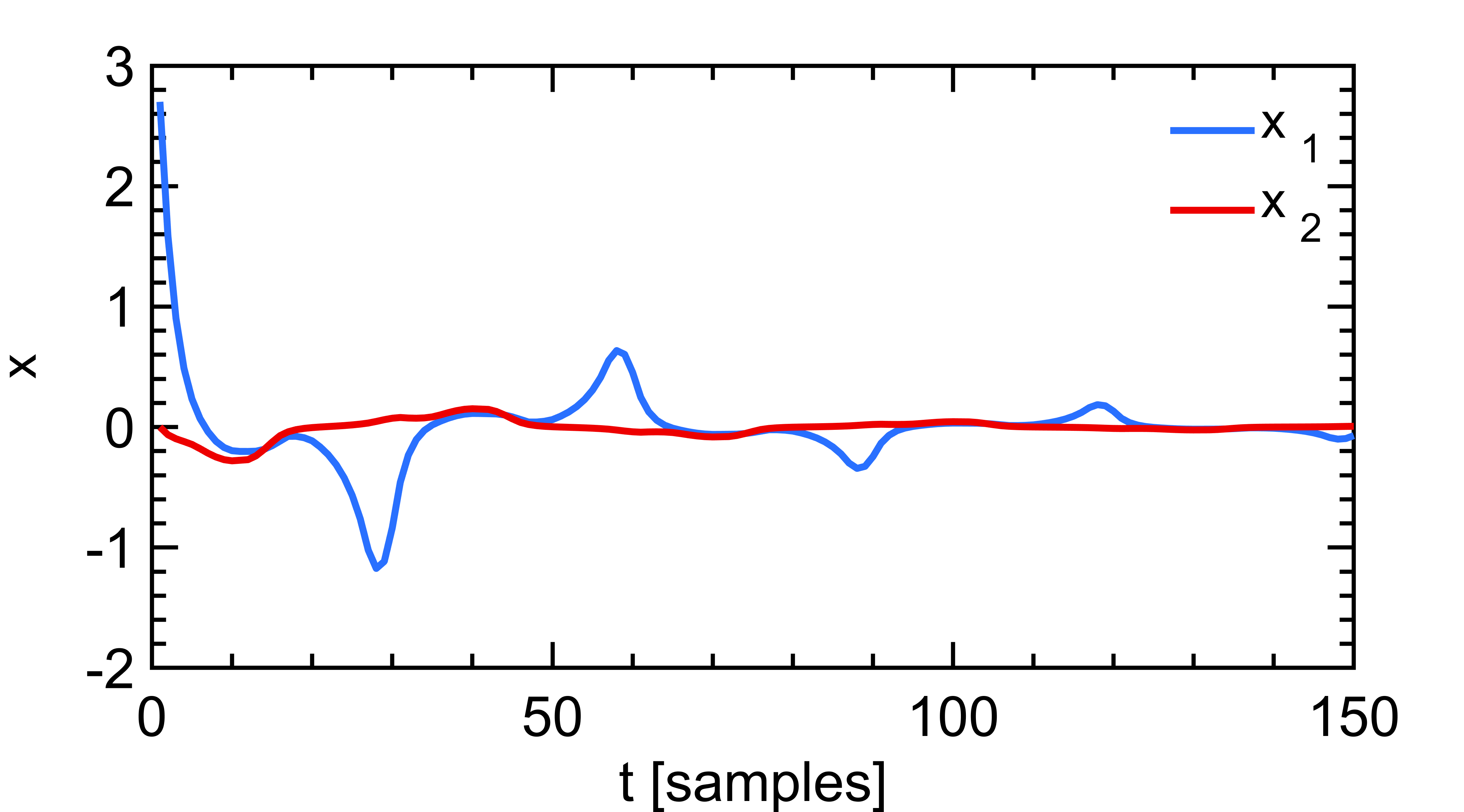}}\label{fig:3cluster}
    \subfigure[Control results using terminal cost and terminal set.]{\includegraphics[width=0.48\textwidth]{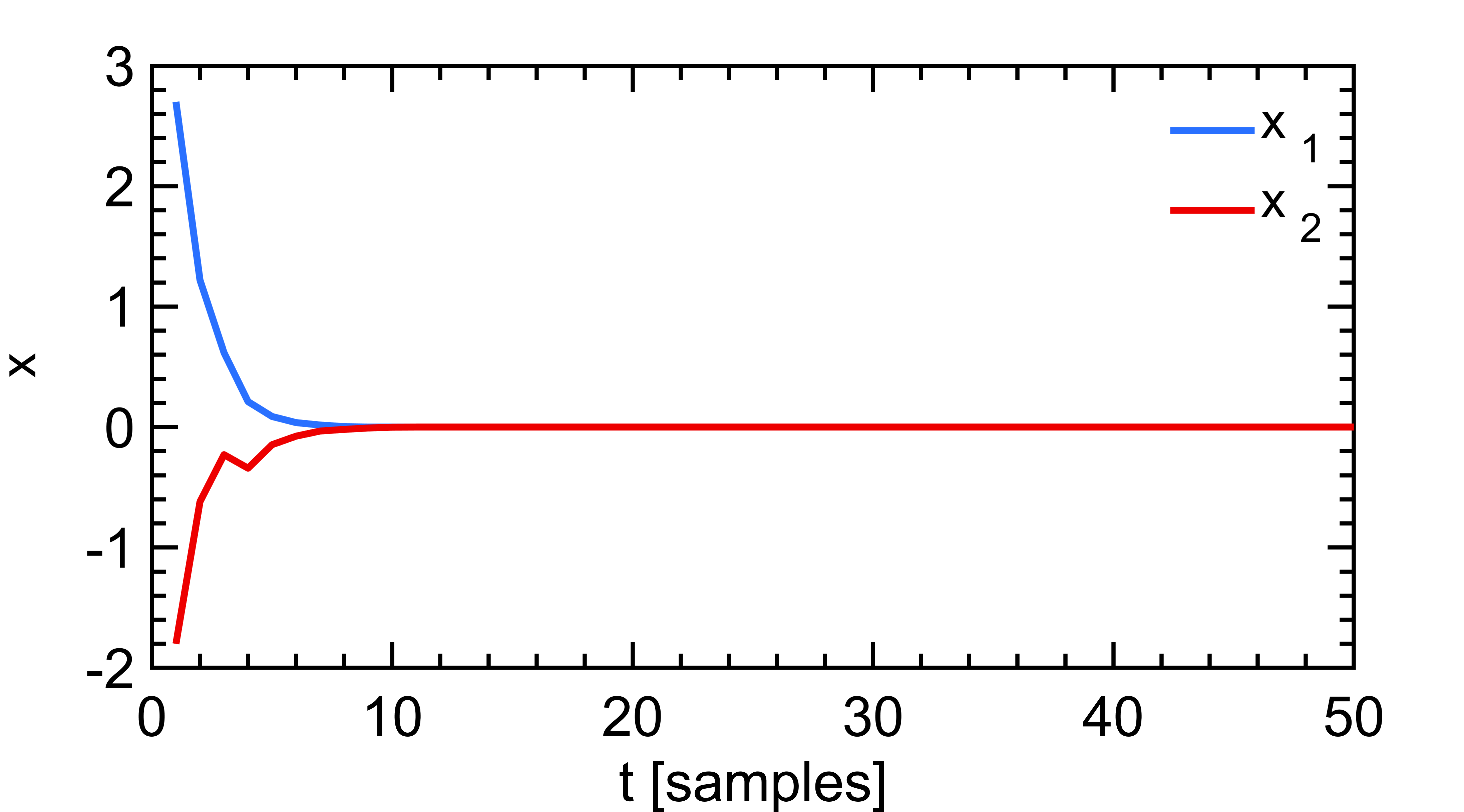}}\label{fig:3cluster_terminal}
\caption{\label{fig:clusterControl}Control results using K-means to generate scenarios.}
\end{figure}

The scheduling signals for control are shown in Fig. \ref{fig:clusterControl}(a), which vary faster than the signals used for model identification in Fig. \ref{fig:data}(a). The control results are shown in Fig. \ref{fig:clusterControl}, where Fig. \ref{fig:clusterControl}(b)-(c) demonstrate that using the terminal cost and terminal set can increase the convergence rate. 

\begin{figure}[H]
\centering
\subfigure{\includegraphics[width=0.48\textwidth]{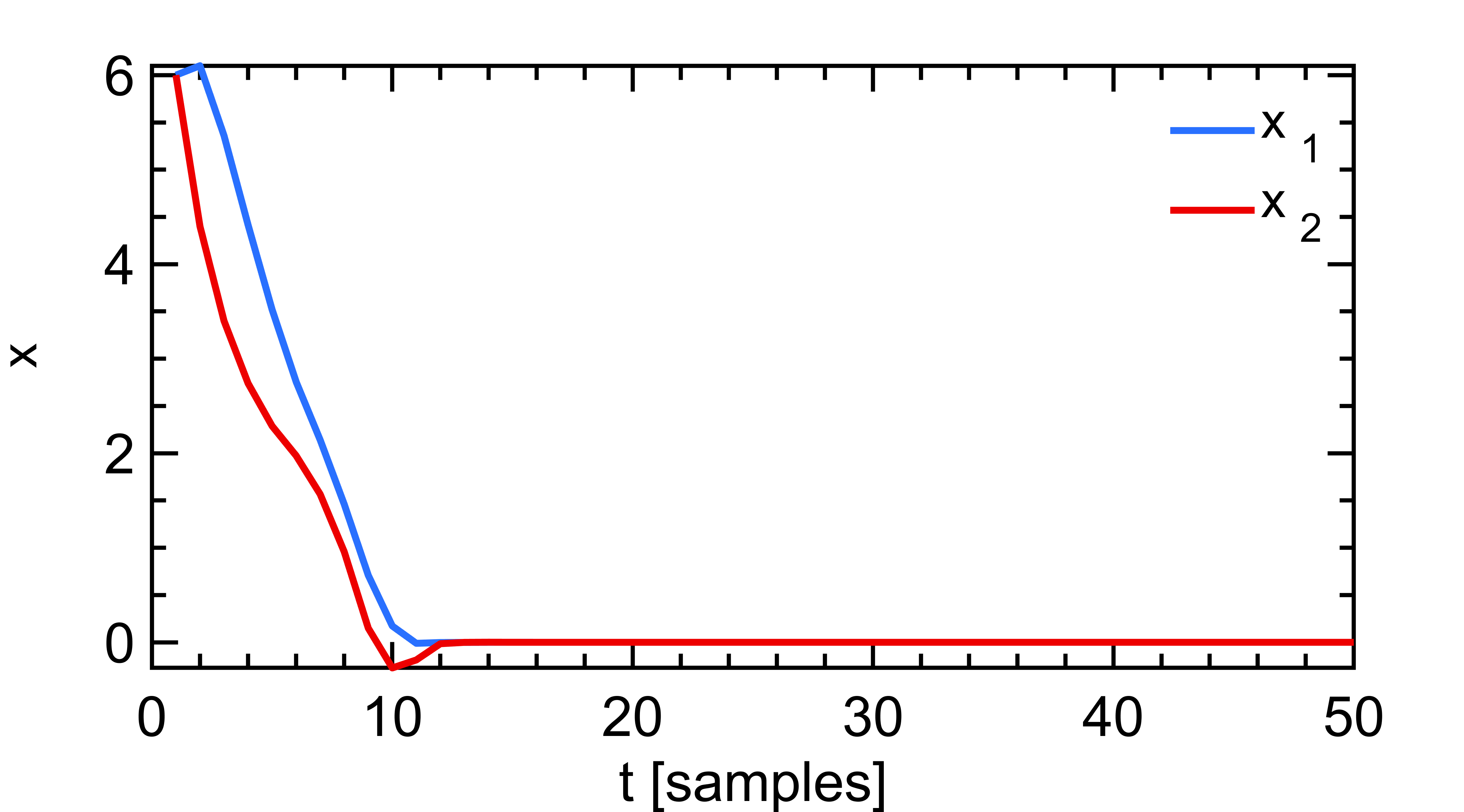}}
\subfigure{\includegraphics[width=0.48\textwidth]{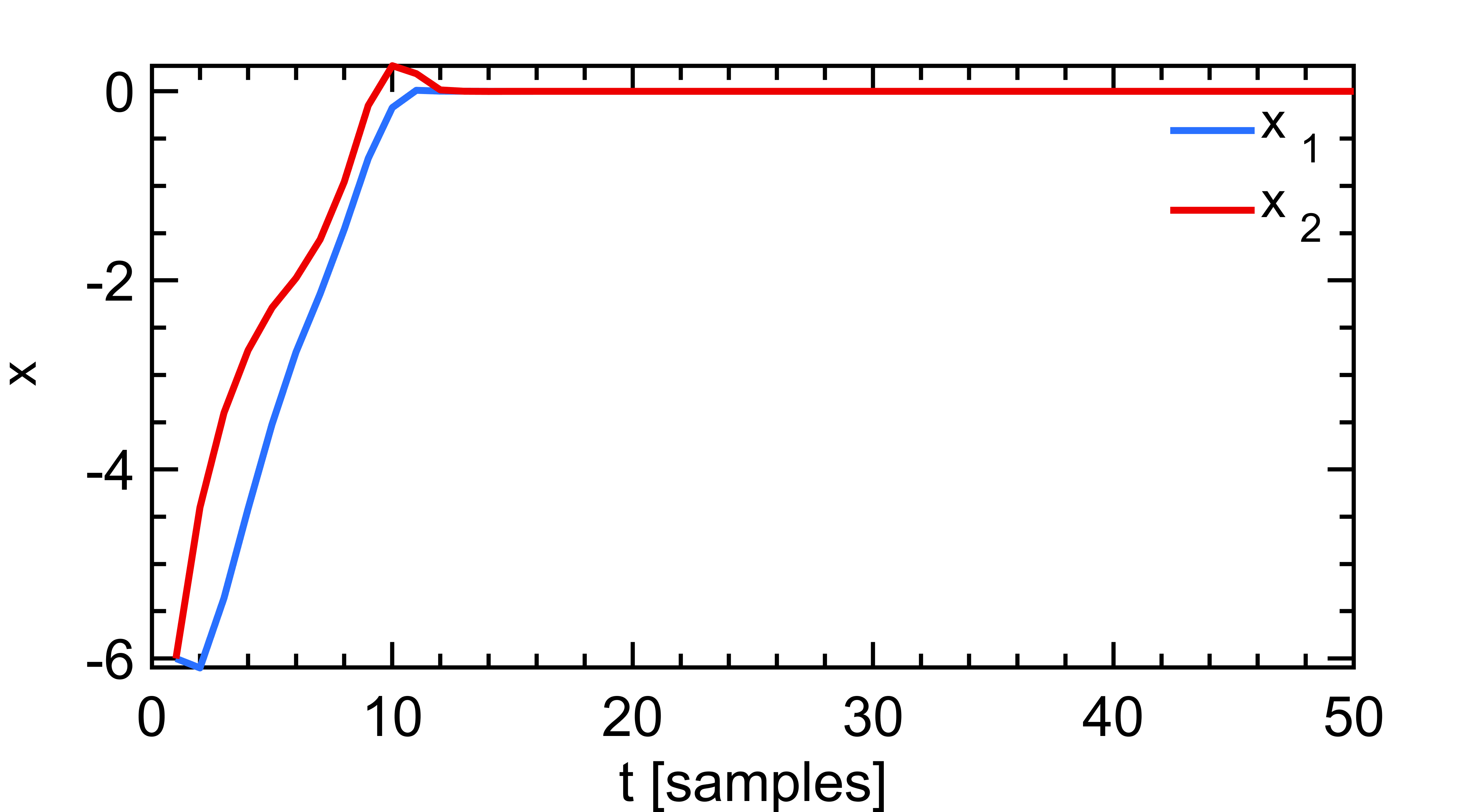}}
\subfigure{\includegraphics[width=0.48\textwidth]{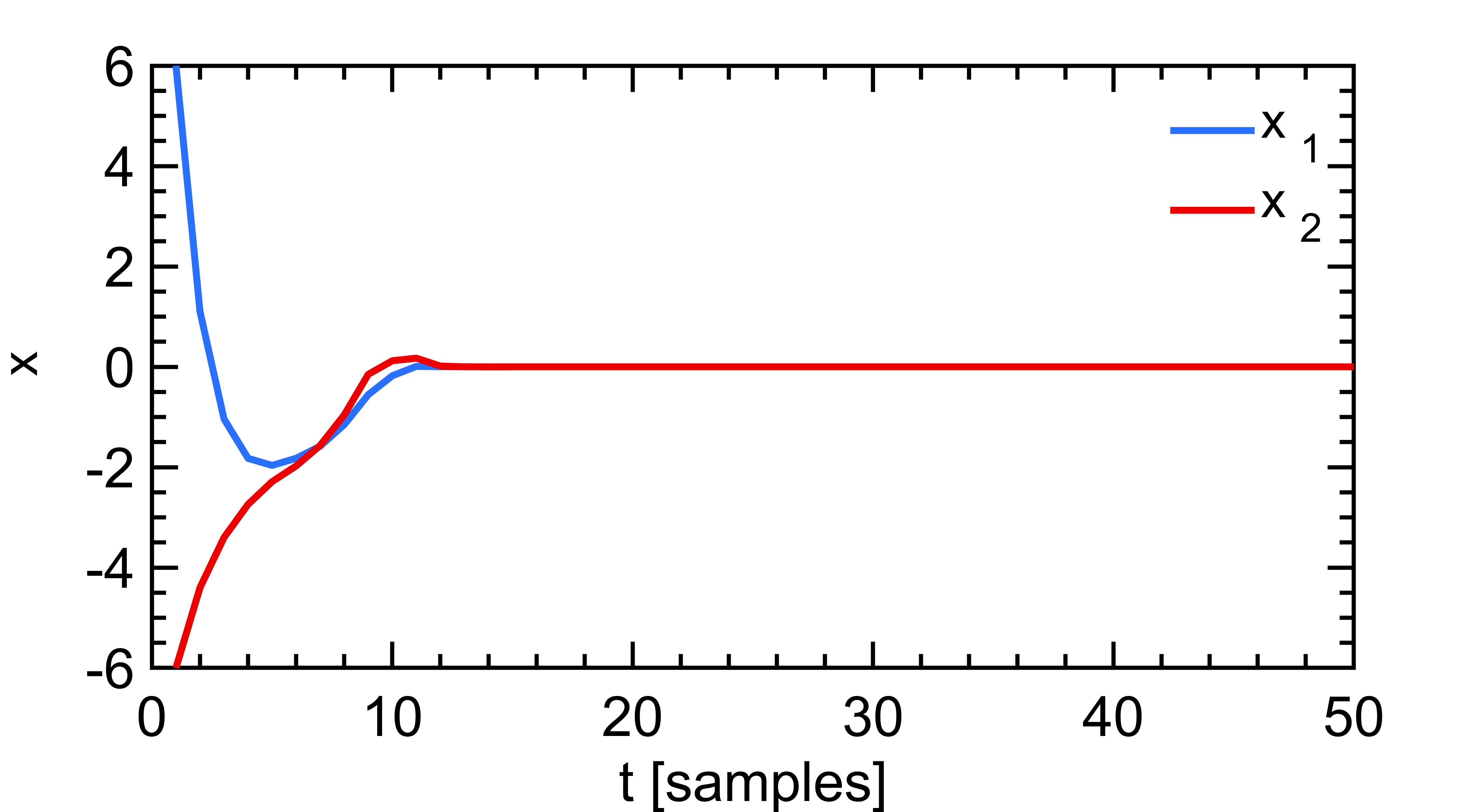}}
\subfigure{\includegraphics[width=0.48\textwidth]{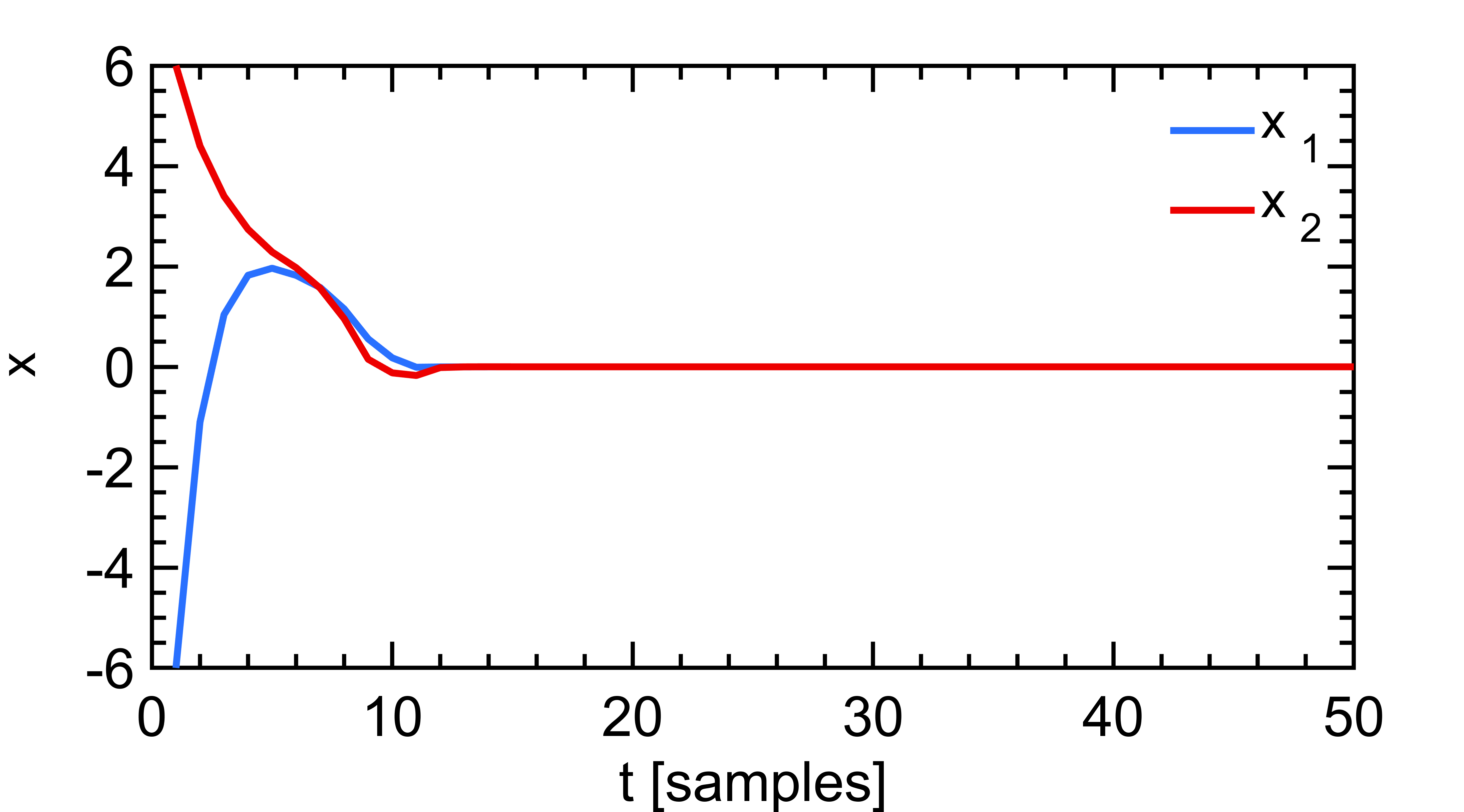}}
\caption{\label{fig:x_v}Control results using K-means to generate scenarios and terminal cost control when the initial states $x_{0}$ are at the vertices of the constraint set $\mathbb{X}$.}
\end{figure}
    
Fig. \ref{fig:x_v} shows that the designed MPC can achieve high control performance even when the initial states are at the vertices of the state constraint set; this is something that was not demonstrated using the approach developed in \citep{hanema2020heterogeneously}. 

\begin{figure}[H]
    \centering
    \subfigure[Random scheduling trajectories.]{\includegraphics[width=0.48\textwidth]{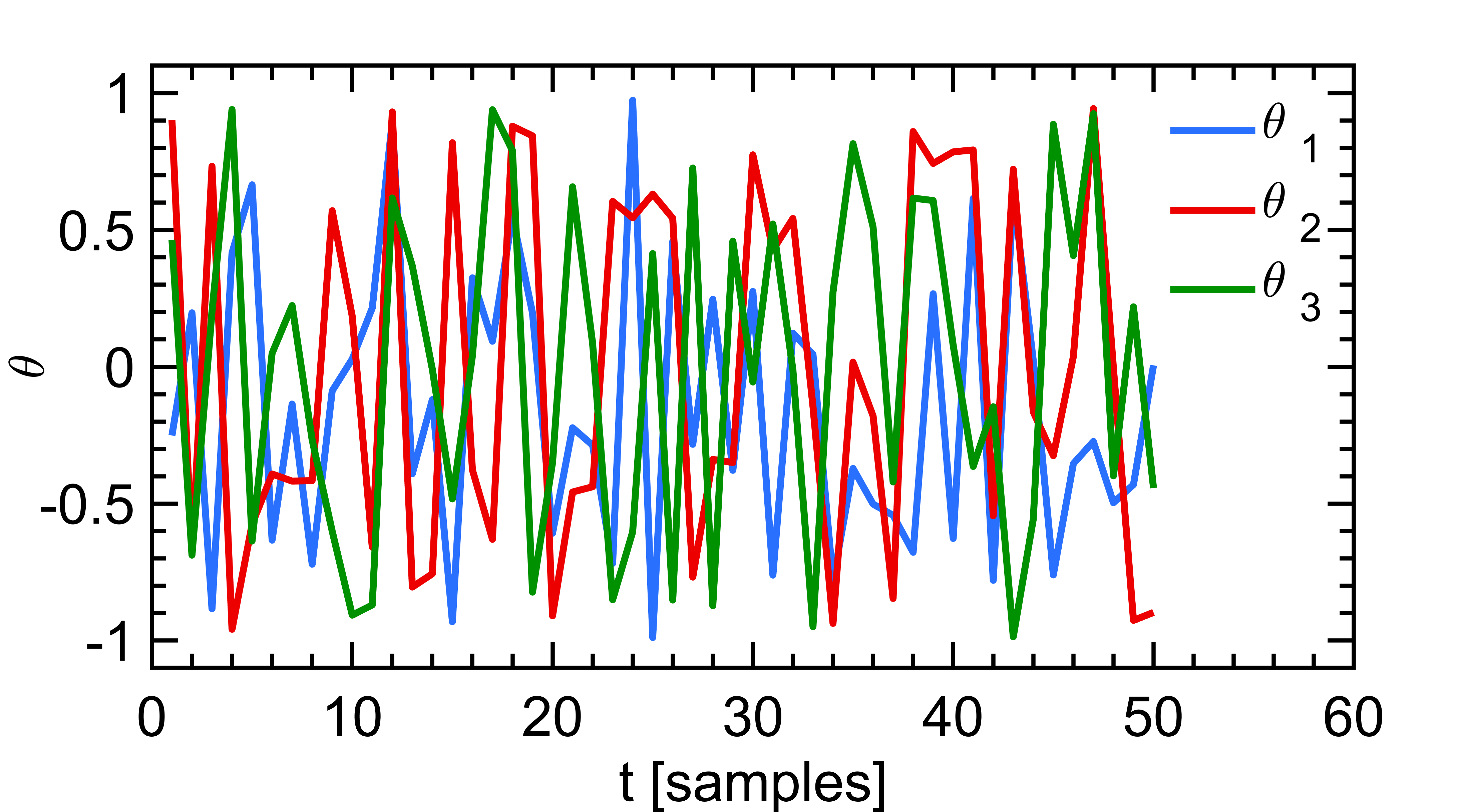}}\label{fig:theta_s_sc}
    \subfigure[State $x_{1}$ trajectory of systems and scenarios.]{\includegraphics[width=0.48\textwidth]{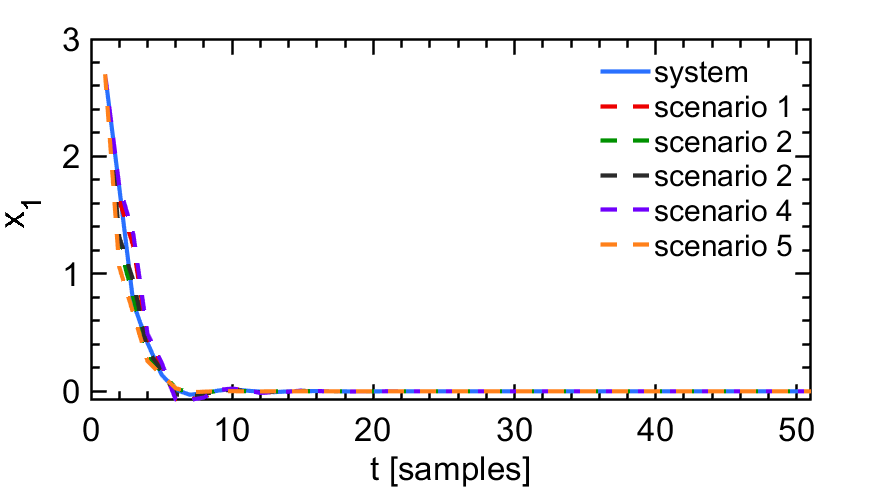}}\label{fig:x1_s_sc}
    \subfigure[State $x_{2}$ trajectory of systems and scenarios.]{\includegraphics[width=0.48\textwidth]{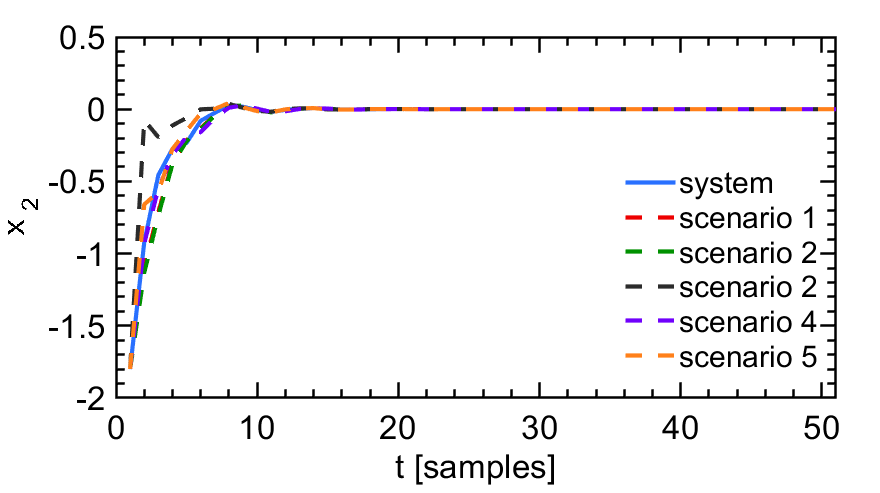}}\label{fig:x2_s_sc}
    \subfigure[Control inputs.]{\includegraphics[width=0.48\textwidth]{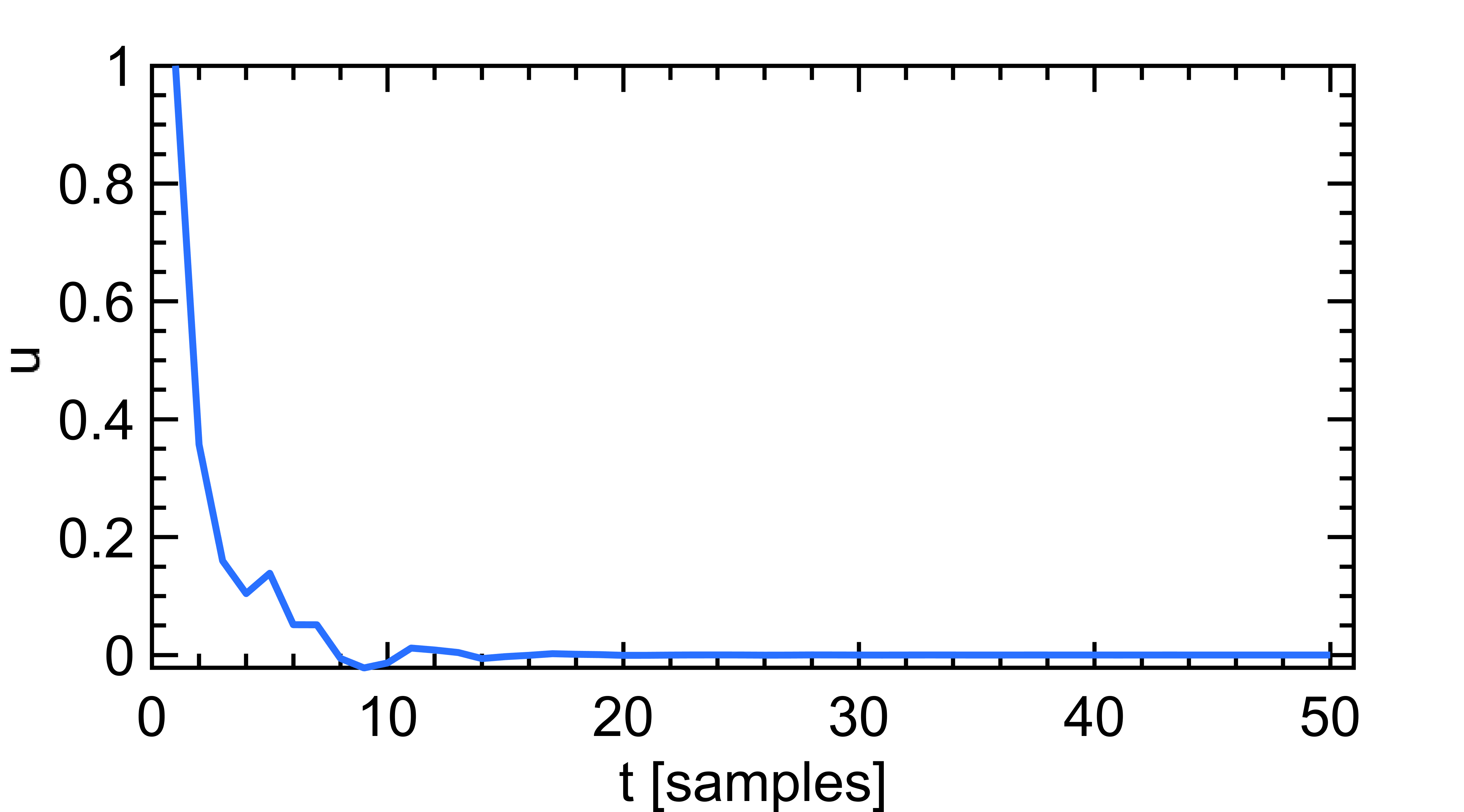}}
    \caption{\label{fig:theta_rand}Control results using terminal cost control when $x_{0}=[2.7;-1.8]$ and a random scheduling trajectory.}

\end{figure}
Additionally, Fig. \ref{fig:theta_rand} shows that the designed MPC is robust against the evolution of the scheduling variables in Fig. \ref{fig:theta_rand}(a) and the real state trajectory is contained among the trajectories of the scenarios (See Fig. \ref{fig:theta_rand}(b)-(c).). 

\subsection{Parameter-varying MIMO System}

The LPV-SS representation of the system is assumed to be
\begin{equation}
\label{eq:lpv-mimo}
    \begin{split}
    x(k+1)&= 
\begin{bmatrix}
\sin(\theta_{1}) & \theta_{1}^{2}+\theta_{1}\theta_{2} \\
\theta_{2}^{3} & \cos(\theta_{1}+\theta_{2}) 
\end{bmatrix}x(k)+ \begin{bmatrix}
\theta_{2}^{4} & \cos(\theta_{2})  \\
\sin(\theta_{1}+\theta_{2}) & \theta_{1}^{3} 
\end{bmatrix}u(k),
    \end{split}
\end{equation}
with constraints and scheduling sets as
\begin{align*}
    &\mathbb{X}  =\{x\in\mathbb{R}^{2}|\|x\|_{\infty}\leq 6\},
    \mathbb{U}  = \{u\in\mathbb{R}^{2}|\left|u\right|_{\infty}\leq 1\} \\
    &\Theta  = \{\theta \in \mathbb{R}^{2}|\|\theta\|_{\infty}\leq 1\}.
\end{align*}
Here, both $A(\cdot)$ and $B(\cdot)$ are nonlinear functions of the scheduling variables.
\subsubsection{Model Identification}\label{sec:sysid} 
We use $\theta_{1}(k)=\sin(0.3k)$ and $\theta_{2}(k)=\sin(0.7k)$ in Fig. \ref{fig:data2}(a) to collect observations $\mathcal{D}=\{(\theta(k), x(k), u(k)), x(k+1)\}$ for model identification purposes. Input signals in Fig. \ref{fig:data2}(b) drawn from the uniform distribution $\mathcal{U}(-0.45,0.45)$ are used to excite the system, and the generated state sequence with initial state $x(0)=[0;0]$ is shown in Fig. \ref{fig:data2}(c). Additionally, $1100$ samples are collected and split into 800 and 300 samples as training and testing sets, respectively.
\begin{figure}[H]
\centering
    \subfigure[Scheduling trajectories.]{\includegraphics[width=0.48\textwidth]{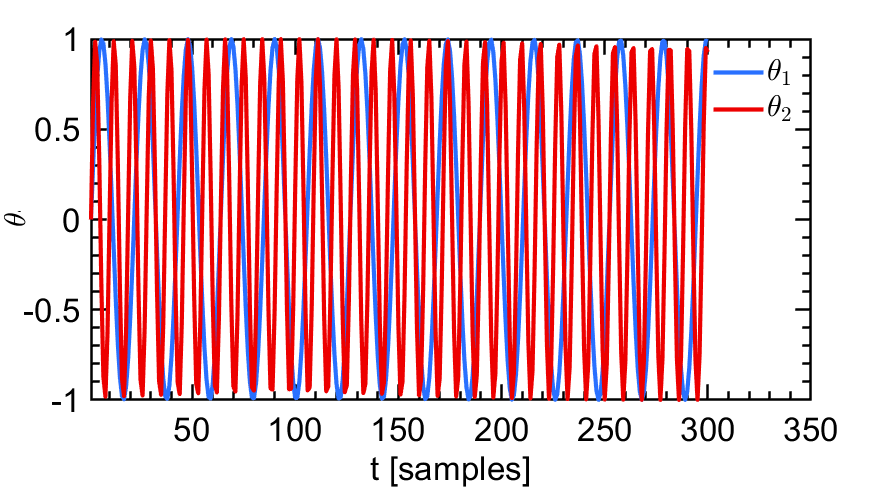}}\label{fig:sched2}
    \vfill
    \subfigure[Inputs to the system.]{\includegraphics[width=0.48\textwidth]{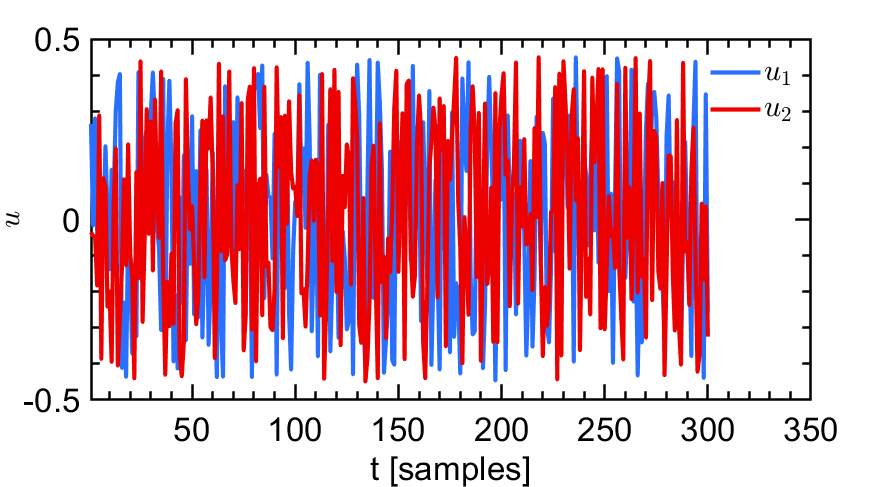}}\label{fig:inp2}
    \subfigure[Sequence of $x$.]{\includegraphics[width=0.48\textwidth]{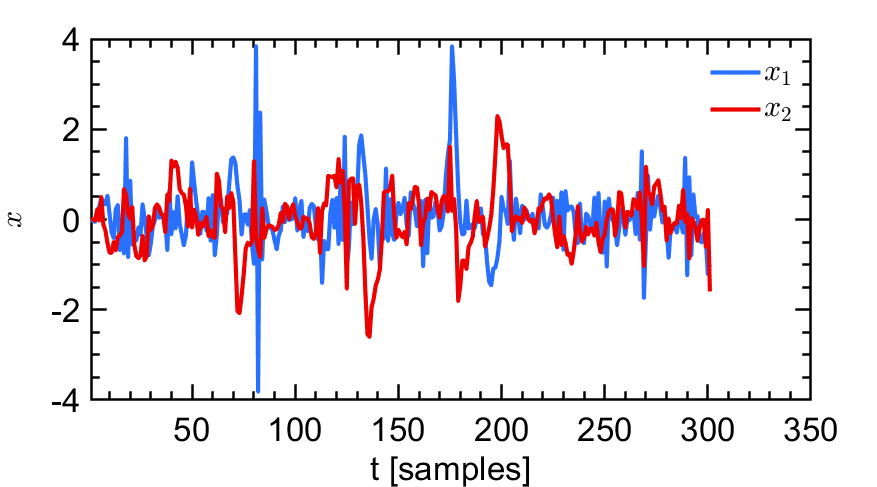}}\label{fig:s_2}
    \caption{\label{fig:data2}Data generated for model identification purposes. For the sake of clarity, only the first 300 training data points are shown here.}
\end{figure}
We use a DenseVariational layer connected to a three-layer fully-connected ANN to represent $A(\cdot)$ and another DenseVariational layer connected to another three-layer fully-connected ANN to represent $B(\cdot)$. All the hidden layers have $32$ hidden units with the Exponential Linear Unit (ELU) activation functions \citep{clevert2015fast} while the output layers have $4$ hidden units without activation functions. The tuning parameters in (\ref{eq:priors}) are determined as $\sigma_{1}=0.3, \sigma_{2}=0.1$. Adam optimizer is used with a learning rate set to $0.001$ and other hyper-parameters as default. Moreover, we first trained an ANN model with the same architecture as the BNN model, used the trained ANN weights to initialize the BNN model, and then trained the BNN model for $10,000$ epochs. The validation results are shown in Fig. \ref{fig:val}. 
\begin{figure}[H]
\centering
    \subfigure{\includegraphics[width=0.48\textwidth]{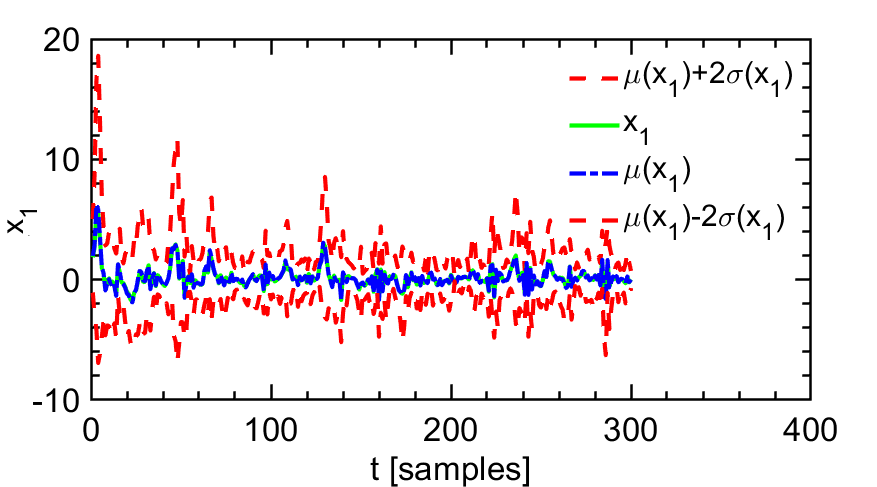}}
    \subfigure{\includegraphics[width=0.48\textwidth]{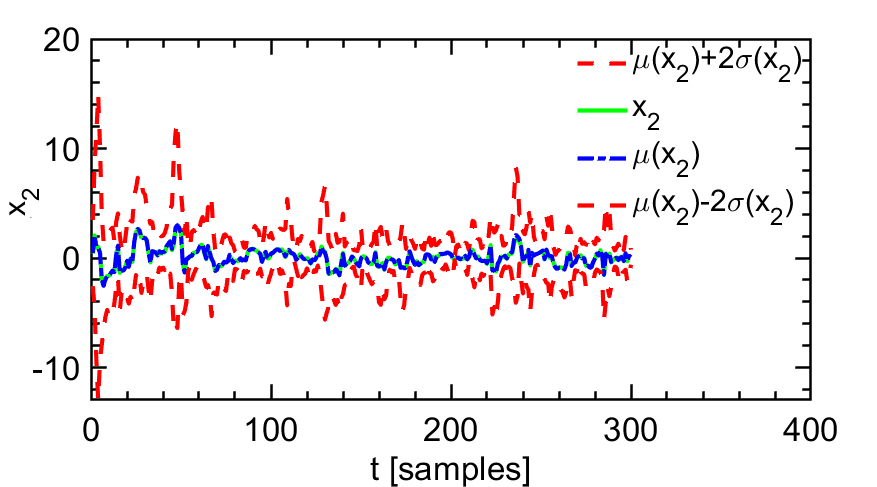}}
    \caption{\label{fig:val}Validation results for the identified BNN model. The $\textup{BFR}=[93.14\%;92.47\%]$ using the estimated mean as predictions for outputs. None of the samples are out of $2\sigma_{x}$.}
\end{figure}

\subsubsection{Validation of The Proposed Approach}

Without assuming extra knowledge on the evolution of the scheduling variables beyond the scheduling sets, we randomly sample $100$ $\theta$'s from the uniform distribution over $\Theta$ and then evaluate both $A(\cdot)$ and $B(\cdot)$ for $N_{\text{MC}}=500$ times using the dynamic functions sampled from the BNN model for each $\theta$. Then, we apply K-means to the concatenations of the vectorized $A$'s and $B$'s to generate the scenarios. The number of clusters is assumed to be 3. Also, $\beta_{\text{M}}=2, \text{M}=A,B$ is considered here for the worst-case scenarios. Therefore, 5 scenarios were used including $\mu_{\text{M}}\pm \beta_{\text{M}}\sigma_{\text{M}}$. Additionally, the scenarios are fixed within the robust horizon of the tree generation in this case due to the limited time-invariant knowledge of $\theta$. The probability of the 5 scenarios is $\mathbf{p}=[0.26; 0.30; 0.26; 0.09; 0.09]$ using the moment matching method. Moreover, in our experiments, $Q=I_{2\times 2},R=I_{2\times 2}$ for the stage cost $\ell$ in (\ref{eq:stagecost}). The prediction horizon is set to 10 and the robust horizon to 1. The RPI set was computed based on the BNN model. 

\begin{figure}[!htbp]
\centering
    \subfigure[RPI set based on the BNN model.]{\includegraphics[width=0.4\textwidth]{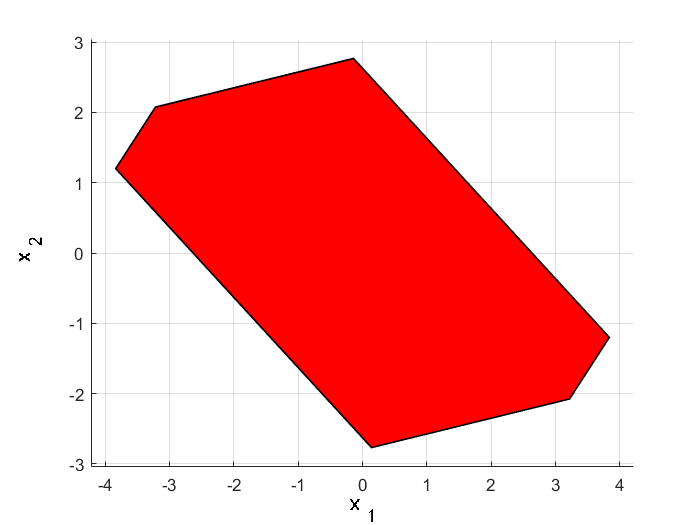}}\label{fig:rpi_2} 
    \subfigure[Random scheduling signals for control.]{\includegraphics[width=0.48\textwidth]{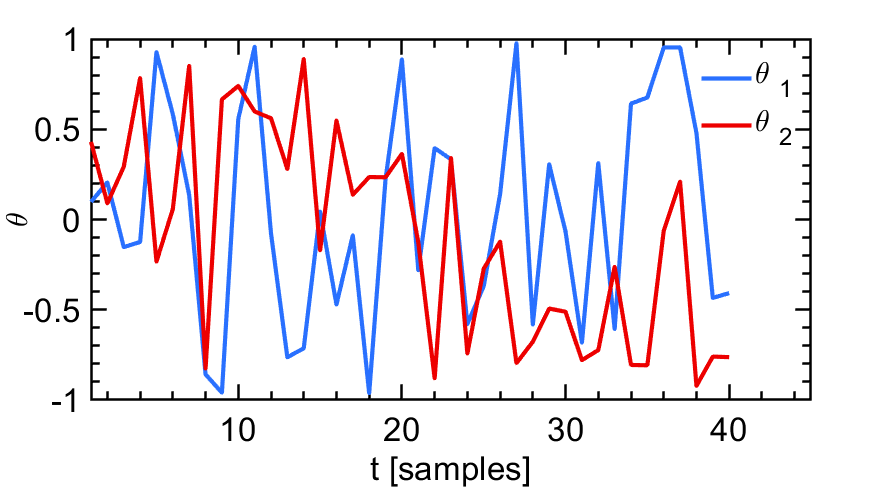}}\label{fig:2_schedule_random}
    \subfigure%[Control results when $x_{0}=[6;6]$.]
    []{\includegraphics[width=0.48\textwidth]{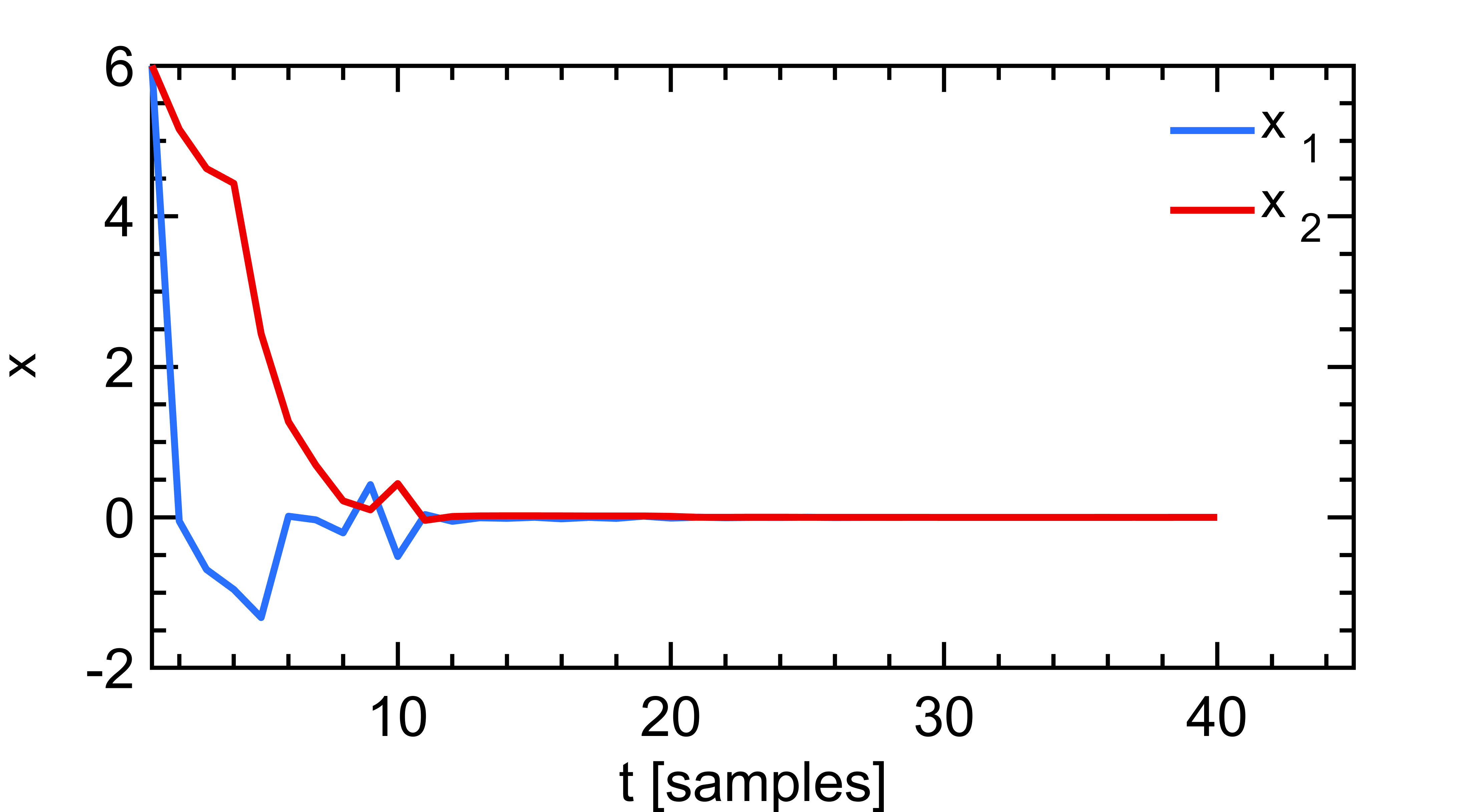}}\label{fig:example_2_66}
    \subfigure%[Control results when $x_{0}=[-6;-6]$.]
    []{\includegraphics[width=0.48\textwidth]{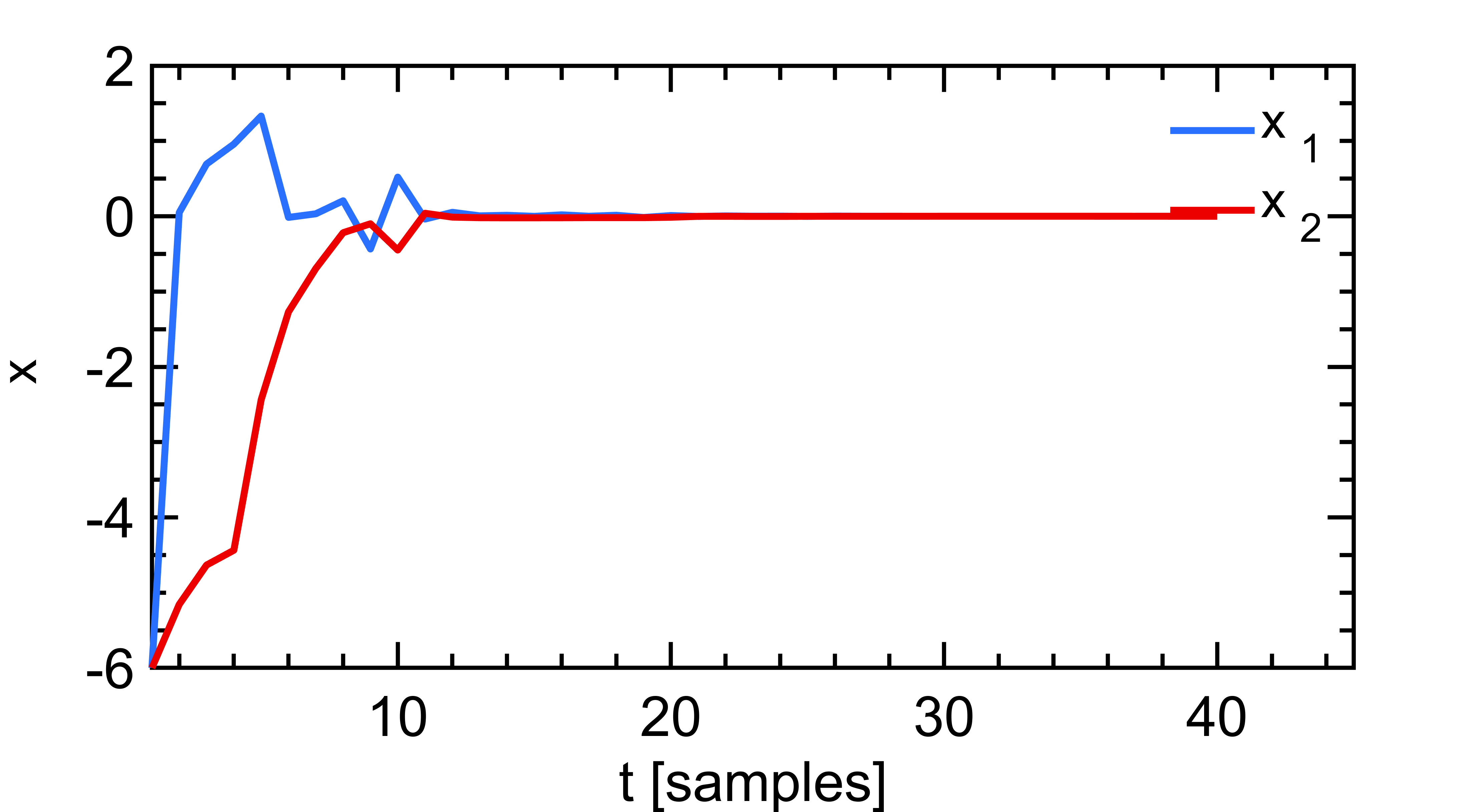}}\label{fig:example_2_-6-6}
    \subfigure%[Control results when $x_{0}=[6;-6]$.]
    []{\includegraphics[width=0.48\textwidth]{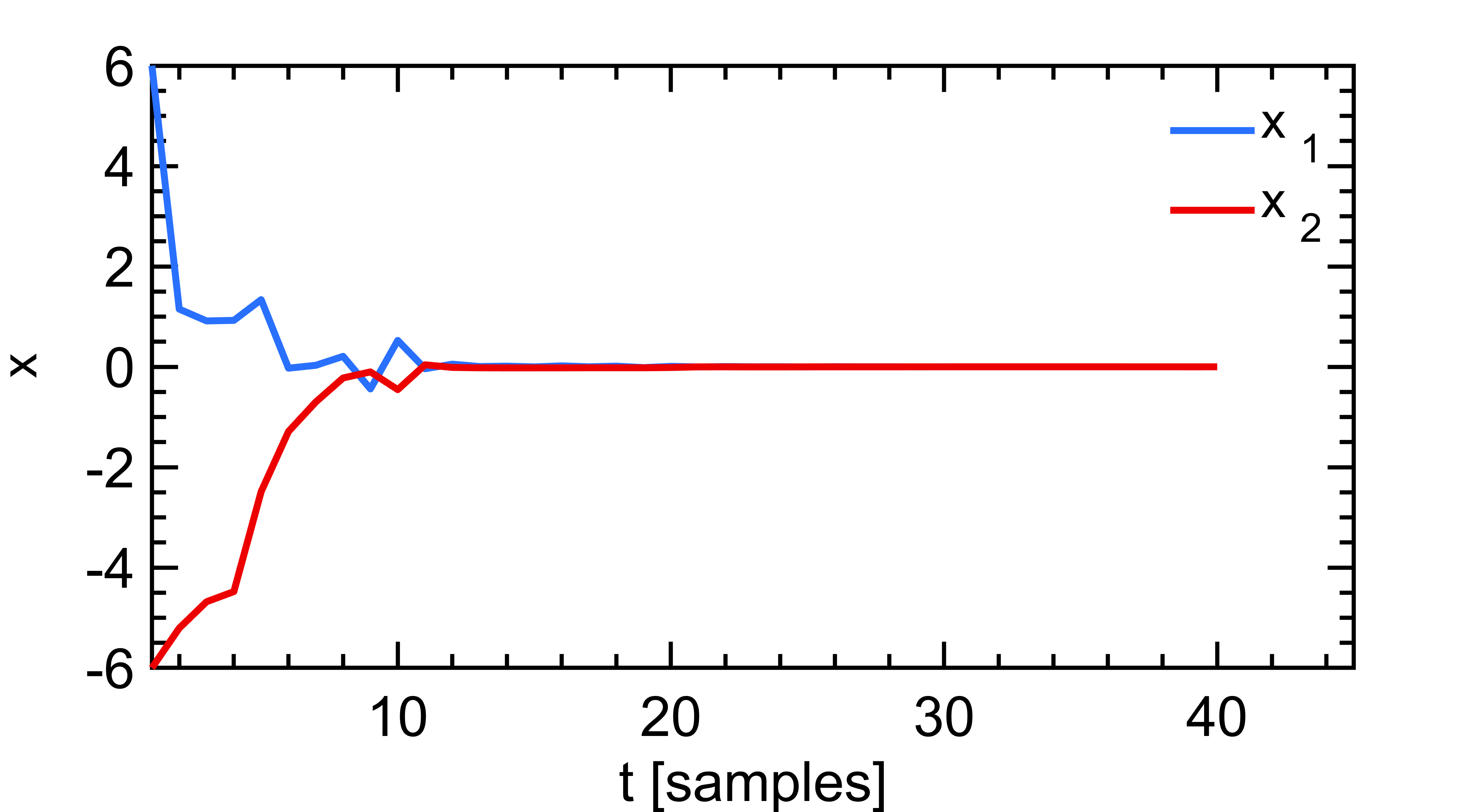}}\label{fig:example_2_6-6}
    \subfigure%[Control results when $x_{0}=[-6;6]$.]
    []{\includegraphics[width=0.48\textwidth]{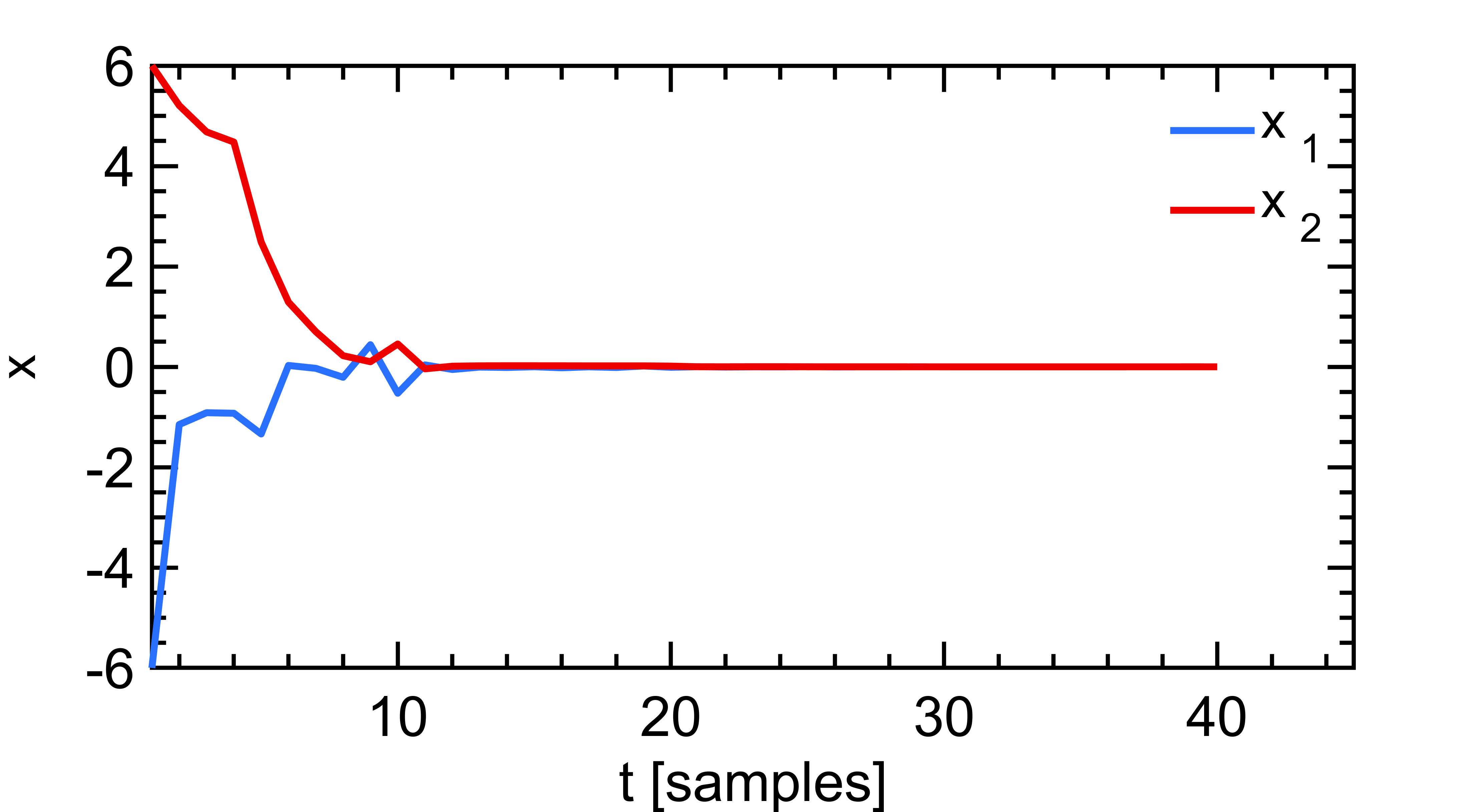}}\label{fig:example_2_-66}
    
    \caption{\label{fig:2_clusterControl}Control results using the proposed approach.}
\end{figure}

\textbf{Results and Discussion:} The computed RPI set based on the BNN model is shown in Fig. \ref{fig:2_clusterControl}(a). The scheduling signal for control is random, as shown in Fig. \ref{fig:2_clusterControl}(b), which varies faster than the signal for identification in Fig. \ref{fig:data2}(a), to demonstrate that the designed MPC is robust against the evolution of the scheduling variable. The control results in Fig. \ref{fig:2_clusterControl}(c-f) show that the designed MPC can achieve good control performance even when the initial states are at the vertices of the state constraint set.

\subsection{Two-tank System}
The cascaded two-tank system \citep{hanema2021stabilizing} can be described by
\begin{subequations} \label{eq:twotank}
\begin{align}
    \rho S_{1} \dot{h}_{1} & = -\rho A_{1}\sqrt{2gh_{1}} + u, \\
    \rho S_{2}\dot{h}_{2} &= \rho A_{1}\sqrt{2gh_{1}}-\rho A_{2}\sqrt{2gh_{2}},
\end{align}
\end{subequations}
where $u$ is the flow of liquid with density $\rho=0.001 ~\textup{kgcm}^{-3}$ pumped into the upper tank. $S_{1} = 2500~\textup{cm}^{2}$, $S_{2} = 1600~\textup{cm}^{2}$, $A_{1} = 9~\textup{cm}^{2}$, and $A_{2} = 4~\textup{cm}^{2}$ denote the cross-sectional areas of the upper tank, the lower tank, the pipe through which the liquid flows into the lower tank, and the pipe through which the liquid flows out, respectively. The control objective is to regulate the levels $h_{1}$ and $h_{2}$ at a given set point. $u$ is available as a control input and subject to the constraint $\mathbb{U} = \{u | 0~\textup{kgs}^{-1}\leq u\leq 4~\textup{kgs}^{-1}\}$. Additionally, the liquid levels satisfy the bounds $\mathbb{X}=\{x=[h_{1},h_{2}]^{\mathrm{T}}|1 ~\textup{cm}\leq h_{1}\leq 35~\textup{cm}, 10~\textup{cm} \leq h_{2} \leq 200~\textup{cm}\}$. The system model \eqref{eq:twotank} is assumed to be unknown for control design and only used for simulation. In the simulation, the goal is to reach a reference value $h_{2}^{*} = 115~\textup{cm}$ of the lower tank. Moreover, the translated state and input variables $\Tilde{x}=x-[22.72,115]^{\mathrm{T}}$ and $\Tilde{u}=u-1.90$ are introduced to convert the problem into a stabilization problem.    
\subsubsection{System Identification}
We apply a random input signal drawn from uniform distribution $U[0,4]$ to collect observations $\mathcal{D}=\{(x(t),u(t)), x(t+1)\}$ for model identification. The sampling time is $0.9$ seconds. The input and the collected state sequences are shown in Fig. \ref{fig:data_twotank}. 
Furthermore, $1000$ samples are collected and split into training and testing sets with a ratio of 65\%/35\%.
\begin{figure}[!htbp]
\centering
    \subfigure[Inputs to the system.]{\includegraphics[width=0.48\textwidth]{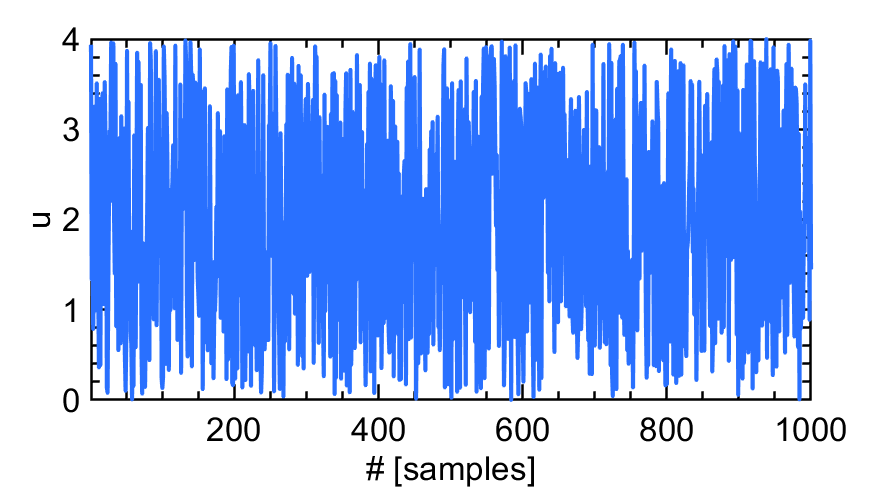}}
    \vfill
    \subfigure[Sequences of $x_{1}=h_{1}$.]{\includegraphics[width=0.48\textwidth]{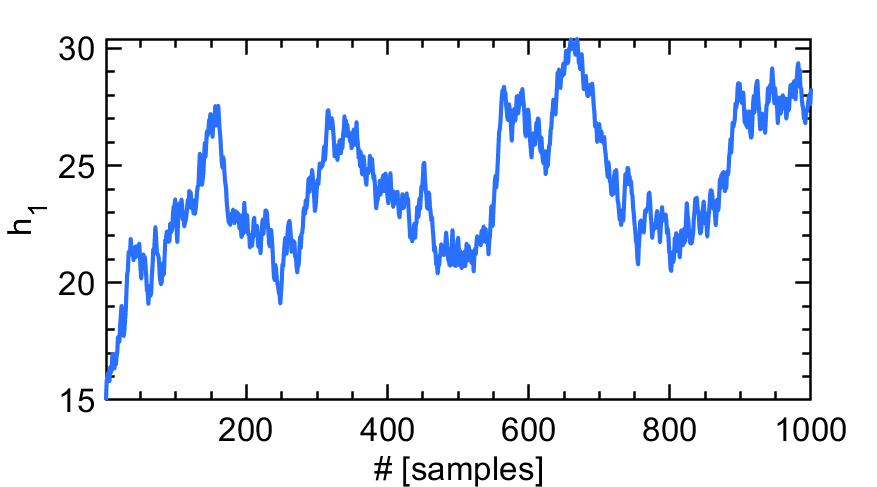}}
    \subfigure[Sequences of $x_{2}=h_{2}$.]{\includegraphics[width=0.48\textwidth]{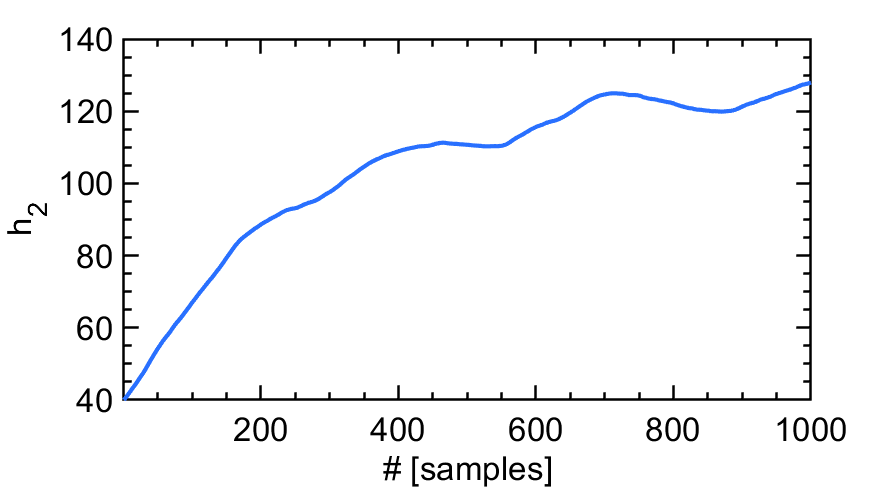}}
    \caption{\label{fig:data_twotank}Data generated for system identification purposes.}
\end{figure}
Since we assume \eqref{eq:twotank} is unknown, we cannot choose the scheduling variables and transform \eqref{eq:twotank} into an exact LPV embedding as \cite{hanema2021stabilizing}, and thus we cannot use the approach in \cite{hanema2021stabilizing} for control design. Instead, we simply use the states as the scheduling variables to learn a model in the form of \eqref{eq:bnn} but treat the scheduling variables as free variables in the prediction horizon of SMPC. In particular, we use a DenseVariational layer connected to a three-layer fully-connected ANN to represent $A(\cdot)$. All the hidden layers have $32$ hidden units with ELU activation functions while the output layers have $4$ hidden units without activation functions. Moreover, we use one Dense layer with $2$ hidden units to represent $B(\cdot)$ and the dense layer does not use bias. The tuning parameters in (\ref{eq:priors}) are determined as $\sigma_{1}=1.5, \sigma_{2}=0.1$. Adam optimizer is used with a learning rate set to $0.001$ and other hyper-parameters as default. Moreover, we first trained an ANN model with the same architecture as the BNN model, used the trained ANN weights to initialize the BNN model, and then trained the BNN model for $50,000$ epochs. The validation results are shown in Fig. \ref{fig:val_twotank}. 
\begin{figure}[H]
\centering
    \subfigure{\includegraphics[width=0.48\textwidth]{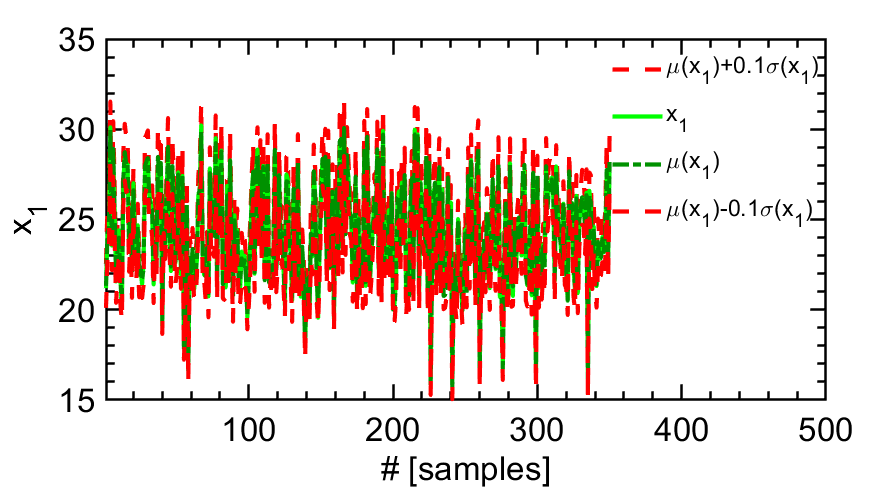}}
    \subfigure{\includegraphics[width=0.48\textwidth]{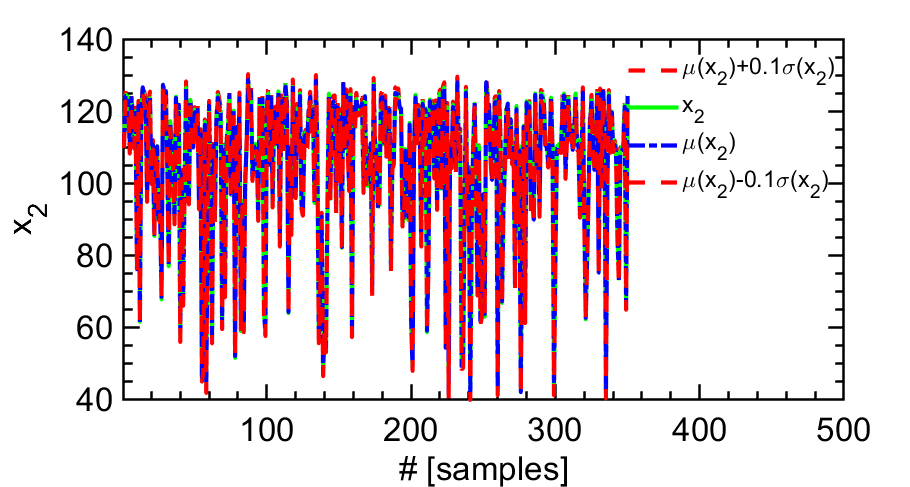}}
    \caption{\label{fig:val_twotank}Validation results for the identified BNN model. The $\textup{BFR}=[92.04\%;97.52\%]$ using the estimated mean as predictions for outputs. 96\% of the states are within 0.1 standard deviations of the average predictions.}
\end{figure}

\subsubsection{Validation of The Proposed Approach}

Without assuming extra knowledge on the evolution of the scheduling variables beyond the scheduling sets, we randomly sample $1000$ $\theta$'s from the uniform distribution over $\Theta=\mathbb{X}$ and then evaluate $A(\cdot)$ for $N_{\text{MC}}=5000$ times using the dynamic functions sampled from the BNN model for each $\theta$. Then, we apply K-means to the vectorized $A$'s to generate the scenarios. The number of clusters is assumed to be 3. Also, $\beta_{A}=0.1$ is considered here for the worst-case scenarios. Therefore, 5 scenarios were used including $\mu_{A}\pm \beta_{A}\sigma_{A}$. Additionally, the scenarios are fixed within the robust horizon of the tree generation in this case due to the limited time-invariant knowledge of $\theta$. The probability of the 5 scenarios is $\mathbf{p}=[0.12; 0.81; 0.07; 0.00; 0.00]$ using the moment matching method. Moreover, in our experiments, $Q=I_{2\times 2},R=10$ for the stage cost $\ell$ in (\ref{eq:stagecost}). The prediction horizon is set to 4 and the robust horizon to 1. Additionally, We use a 4-layer fully-connected NN with 4 and 8 units in the 2 hidden layers to model the coordinate transformation from the scheduling variable in \eqref{eq:bnn} to the scheduling variable in \eqref{eq:lpv-affine}. The RPI set shown in Figure \ref{fig:twotank_clusterControl} (a) was computed based on the BNN model. 
\begin{figure}[!htbp]
\centering
    \subfigure[RPI set based on the BNN model.]{\includegraphics[width=0.4\textwidth]{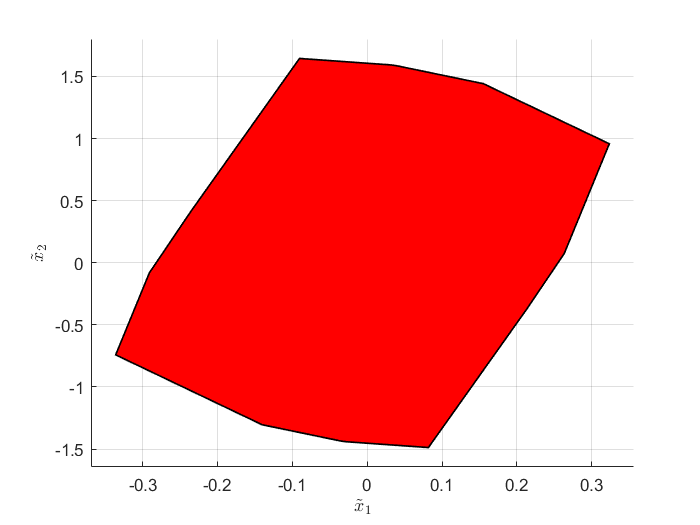}}
    \subfigure
    [$x_{1}$ profile.]{\includegraphics[width=0.48\textwidth]{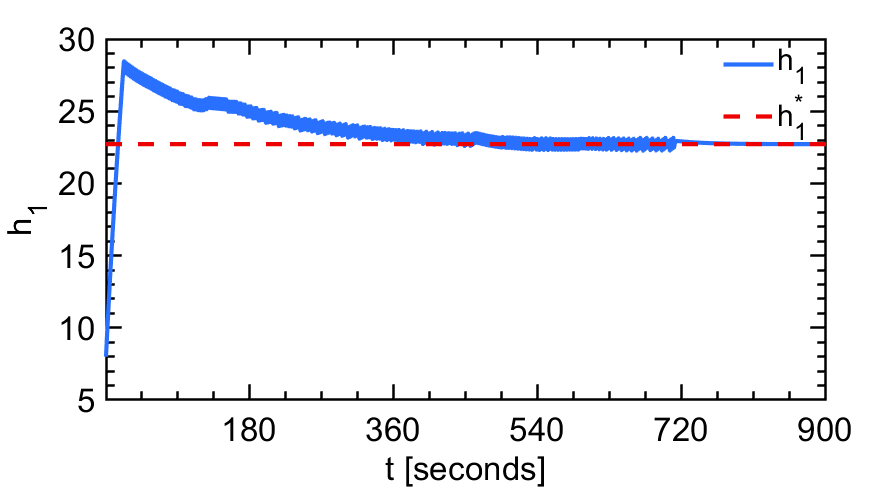}}
    \subfigure
    [$x_{2}$ profile.]{\includegraphics[width=0.48\textwidth]{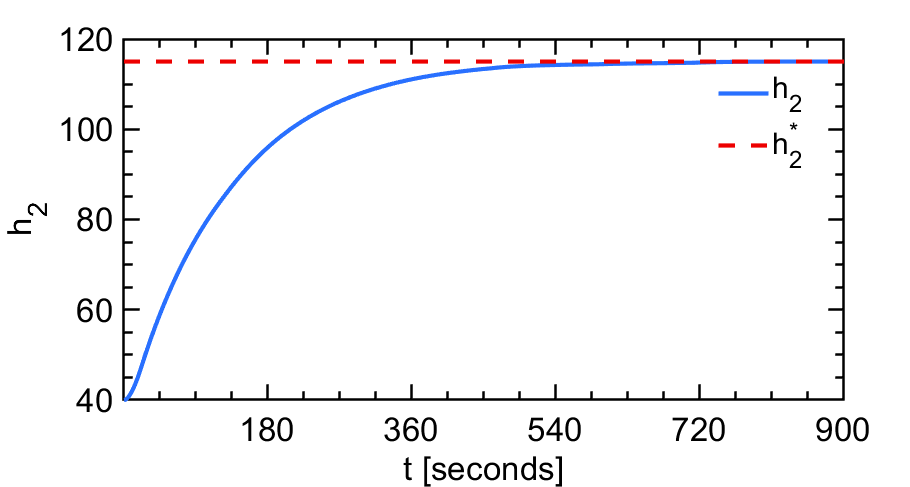}}
    \subfigure[Control inputs.]
    {\includegraphics[width=0.48\textwidth]{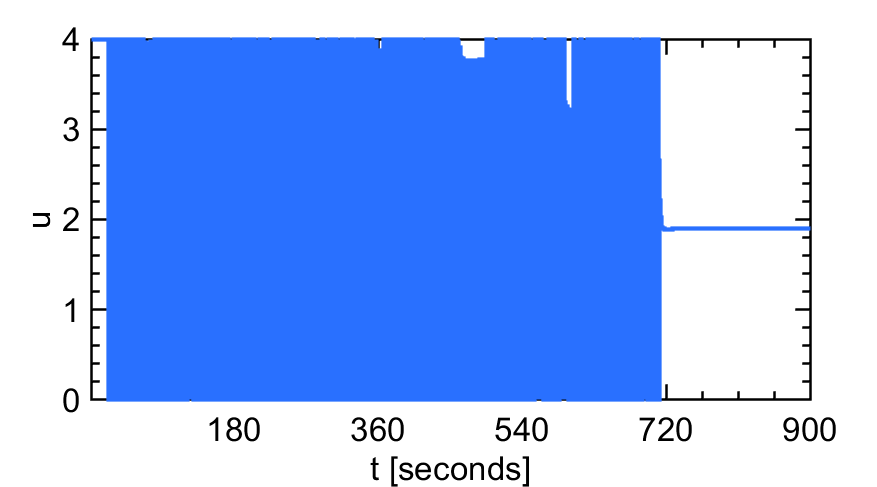}}
    \caption{\label{fig:twotank_clusterControl}Control results using the proposed approach.}
\end{figure}

The control results in Fig. \ref{fig:twotank_clusterControl}(c) show that the designed SMPC can bring the liquid level $h_{2}$ of the lower tank to the reference value while satisfying the system constraints. The control inputs in Fig. \ref{fig:twotank_clusterControl}(d) fluctuate between the limits in the early stages of the control process, which may result from the conservatives of the BNN model and the generated scenarios. Moreover, it is noted that the proposed approach reached the set point slower than the approach that assumes a known system model and uses the exact LPV embedding in \cite{hanema2021stabilizing}, as the data-driven model can be conservative, compared with the exact LPV model of the system. However, the data-driven model can be refined using the closed-loop data to improve the control performance, which will be investigated in the future work.  

% 616 seconds vs 300
%%%%%%%%%%%%%%%%%%%%%%%%%%%%%%%%%%%%%%%%%%%%%%%%%%%%%%%%%%%%%%%%%%%%%%
\section{Concluding Remarks}

In this paper, a learning-based MPC design approach was proposed for systems described in the LPV framework. BNNs were used to learn from input-output data an LPV-SS model with epistemic uncertainty quantification. Then, the epistemic uncertainty from the system identification and imprecise knowledge of the future scheduling variables were jointly considered for control design with safety guarantees. SMPC was proposed to consider safety when generating scenarios. K-means clustering and moment matching were used to generate scenarios with probabilities that can retain the stochastic properties of the joint uncertainty of the model and the scheduling variables. To guarantee closed-loop stability, parameter-dependent terminal cost, and controller were designed, which can improve the control performance, together with a terminal RPI set. Numerical experiments and simulations were used to show that the proposed approach can ensure safety and achieve the desired control performance.

In our future work, we plan to consider the effects of measurement noise of scheduling variables on the proposed approach, as exact measurements of these parameters can be impractical in real applications. Moreover, we will improve the proposed approaches to evaluate the probabilistic safety of BNN models and develop online adaptation approaches to reduce the conservativeness of BNN models using closed-loop data. 

%%%%%%%%%%%%%%%%%%%%%%%%%%%%%%%%%%%%%%%%%%%%%%%%%%%%%%%%%%%%%%%%%%%%%%

\section*{Disclosure Statement}

No potential conflict of interest was reported by the authors.

\section*{Funding}

This work was financially supported by the United States National Science Foundation under award \#1912757. The second author's work was funded by the Deutsche Forschungsgemeinschaft (DFG, German Research Foundation) under project \#419290163.
\bibliographystyle{apacite}
\bibliography{ref}

\end{document}